\documentclass{article}

\usepackage{arxiv}

\usepackage{amsmath,amsfonts,amssymb, amsthm}

\usepackage[usenames,dvipsnames]{xcolor}
\definecolor{niceRed}{RGB}{190,38,38}
\definecolor{niceBlue}{HTML}{0466a7}

\usepackage[utf8]{inputenc} % allow utf-8 input
\usepackage[T1]{fontenc}    % use 8-bit T1 fonts
\usepackage{hyperref}       % hyperlinks
\hypersetup{
	colorlinks = true,
	urlcolor = {black},
	linkcolor = {black},
	citecolor = {niceRed} % was nicePink
}

\usepackage{natbib}
\bibpunct[, ]{(}{)}{,}{a}{}{,}%
\def\BIBand{and}%

\usepackage{xurl}           % simple URL typesetting
\usepackage{booktabs}       % professional-quality tables
\usepackage{multirow}
\usepackage{nicefrac}       % compact symbols for 1/2, etc.
\usepackage{microtype}      % microtypography
\usepackage{cleveref}       % smart cross-referencing
\usepackage{graphicx}
\usepackage{doi}
\usepackage{tikz}
\usepackage{mathtools}
\usepackage{pifont} 		% for cross and check mark

\usepackage{subcaption} % for subfigures
\usepackage{booktabs}	% nice tables
\usepackage{multirow}
\usepackage{algpseudocode}
\usepackage[ruled]{algorithm2e}
\usepackage{float}
\usepackage{enumitem}
\usepackage{tocloft}

\newtheoremstyle{spaced}%
  {10pt}   % Space above
  {10pt}   % Space below
  {\itshape} % Body font
  {}       % Indent amount
  {\bfseries} % Theorem head font
  {.}      % Punctuation after theorem head
  { }      % Space after theorem head
  {\thmname{#1}\thmnumber{ #2}\thmnote{ (#3)}} % Theorem head spec
\makeatother

\theoremstyle{spaced}
\newtheorem{definition}{Definition}[section]
\newtheorem{lemma}{Lemma}[section]

\newtheorem{theorem}{Theorem}[section]
\newtheorem{proposition}{Proposition}[section]
\newtheorem{corollary}{Corollary}[section]
\newtheorem{remark}{Remark}[section]

\definecolor{lightgray}{gray}{0.7}
\definecolor{synblue}{RGB}{0,66,118}
\definecolor{synred}{RGB}{143,53,71}
\definecolor{syngreen}{RGB}{34,143,156}
\definecolor{darkviolet}{rgb}{0.58, 0.0, 0.83}

\definecolor{customblue}{HTML}{4C72B0}
\definecolor{customred}{HTML}{C44E52}

\makeatletter

\newcommand{\Acal}{\mathcal{A}}
\newcommand{\Bcal}{\mathcal{B}}
\newcommand{\Ecal}{\mathcal{E}}

\newcommand{\Ical}{\mathcal{I}}

\newcommand{\Mcal}{\mathcal{M}}
\newcommand{\Ncal}{\mathcal{N}}

\newcommand{\Ucal}{\mathcal{U}}
\newcommand{\Vcal}{\mathcal{V}}
\newcommand{\Xcal}{\mathcal{X}}
\newcommand{\R}{\mathbb{R}}

\DeclareMathOperator{\E}{\mathbb{E}}

\DeclareMathOperator{\proj}{\Pi}

\DeclareMathOperator*{\argmin}{argmin}

\newcommand{\norm}[1]{\left\| #1 \right\|}
\newcommand{\rr}{\right)}
\renewcommand{\ll}{\left(}
\newcommand{\game}{\mathcal{G}} % Game
\newcommand{\nfgdiscr}{\game = \ll \players, \pures, \pay \rr} % Full (discrete) game
\newcommand{\players}{\mathcal{N}} % Space of players
\newcommand{\n}{n} % Number of players
\newcommand{\mi}{{-i}} % Other players

\newcommand{\pures}{\mathcal{A}} % Space of pure strategy profiles
\newcommand{\pay}{u} % Payoffs
{\theoremstyle{THkey}}
\title{Agentic Markets: Game Dynamics and Equilibrium in Markets with Learning Agents}

\newif\ifuniqueAffiliation
\uniqueAffiliationtrue

\author{
	Martin Bichler \\
	\small Department of Computer Science\\
	\small Technical University of Munich\\
	\small \texttt{m.bichler@tum.de} \\
	\And
	Julius Durmann\\
	\small Department of Computer Science\\
	\small Technical University of Munich \\
	\small \texttt{julius.durmann@tum.de}\\
	\And
	Matthias Oberlechner\\
	\small Department of Computer Science\\
	\small Technical University of Munich \\
	\small \texttt{matthias.oberlechner@tum.de}\\
}

\renewcommand{\headeright}{}
\renewcommand{\undertitle}{}
\renewcommand{\shorttitle}{}

\hypersetup{
	pdftitle={Agentic Markets: Game Dynamics and Equilibrium in Markets with Learning Agents},
	pdfsubject={cs.GT},
	pdfauthor={Bichler et al.},
	pdfkeywords={} ,
}

\begin{document}
\maketitle
\pagenumbering{roman} 

\begin{abstract}
Autonomous and learning agents increasingly participate in markets - setting prices, placing bids, ordering inventory. Such agents are not just aiming to optimize in an uncertain environment; they are making decisions in a game-theoretical environment where the decision of one agent influences the profit of other agents. While game theory usually predicts outcomes of strategic interaction as an equilibrium, it does not capture how repeated interaction of learning agents arrives at a certain outcome. This article surveys developments in modeling agent behavior as dynamical systems, with a focus on projected gradient and no-regret learning algorithms. In general, learning in games can lead to all types of dynamics, including convergence to equilibrium, but also cycles and chaotic behavior. It is important to understand when we can expect efficient equilibrium in automated markets and when this is not the case. Thus, we analyze when and how learning agents converge to an equilibrium of a market game, drawing on tools from variational inequalities and Lyapunov stability theory. Special attention is given to the stability of projected dynamics and the convergence to equilibrium sets as limiting outcomes. Overall, the paper provides mathematical foundations for analyzing stability and convergence in agentic markets driven by autonomous, learning agents.
\end{abstract}

%\vspace{1.5cm}
\newpage
\setlength{\cftbeforesecskip}{6pt}  % Remove space before section entries
\setlength{\cftbeforesubsecskip}{4pt}  % Remove space before subsection entries (if needed)
\setcounter{tocdepth}{2}
\tableofcontents

\newpage
\pagenumbering{arabic}

%%%%%%%%%%%%%%%%%%%%%%%%%%%%%%%%%%%%%%%%%%%%%%%%%%%%%%%%%%%%%%%%%%%%%%%%%%%%%%%%%%%%%%%%%%%%%%%%%%%%%%%%%%
%	1  INTRODUCTION 					                                                                 %
%%%%%%%%%%%%%%%%%%%%%%%%%%%%%%%%%%%%%%%%%%%%%%%%%%%%%%%%%%%%%%%%%%%%%%%%%%%%%%%%%%%%%%%%%%%%%%%%%%%%%%%%%%

\section{Introduction}

Markets increasingly rely on autonomous agents - algorithmic systems that make pricing, bidding, and ordering decisions with minimal or no human intervention \citep{acharya2025agentic}. From dynamic pricing on online retail platforms to bidding in ad auctions and order allocation in supply chains, agentic AI systems are becoming pervasive \citep{economist2018, fountaine2019building, sharma2022role}. These agents operate in environments characterized by repeated strategic interactions, which are best modeled as games. We will refer to such environments as \textit{agentic markets}.

While traditional game theory focuses on equilibrium outcomes, it assumes that agents are rational and play an equilibrium strategy from the start. However, in real-world markets, artificial agents would not even have the necessary information about the costs or values of others to derive an equilibrium strategy. Instead, there is evidence that online learning algorithms \citep{shalev-shwartz2007OnlineLearningTheory} are widely used to act in such environments and set prices or submit bids in display ad auctions. Game theory offers limited insight into the \textit{dynamics} of how these outcomes are reached (if at all) by learning agents. 
A growing body of experimental and theoretical work shows that learning dynamics can result in complex behaviors such as cycles, divergence, or even chaos \citep{mertikopoulos2018cycles, bailey2018multiplicative, cheung2020chaos}. \citet{sanders2018prevalence} argues that complex non-equilibrium behavior, exemplified by chaos, may be the norm in complicated games with many players. Consequently, the mere existence of a Nash equilibrium does not guarantee that agentic markets will implement equilibrium outcomes.

This gap calls for a dynamical systems perspective on agentic markets. To understand how algorithmic agents interact, adapt, and converge (or not) to equilibrium, we must go beyond static equilibrium concepts. In particular, we require tools to characterize and certify the stability of learning dynamics, especially under constraints such as those arising from probability simplices in mixed extensions of finite, normal-form games. The latter are important, because every continuous game can be discretized into a game with a finite number of actions. 

In this paper, we develop a framework for analyzing game dynamics in agentic markets that is thoroughly based on dynamical systems theory. We show how gradient-based algorithms such as projected gradient ascent can be modeled as dynamical systems on constrained domains. Projected gradient ascent belongs to the class of follow-the-regularized-leader (FTRL) algorithms, which are practically relevant and lend themselves to a formal analysis.  

It cannot be expected that uncoupled dynamics lead to Nash equilibrium in all games, and there are several counterexamples \citep{daskalakis2010learning, flokas2020no, hart2003uncoupled, milionis2022nash}. Therefore, it is important to study specific models of agentic markets such as oligopoly competition or auctions, and understand whether learning algorithms can be expected to converge in specific environments. 

Of particular interest in this study is the use of Lyapunov functions to certify stability. While classical Lyapunov theory applies to smooth, unconstrained systems, markets with autonomous agents often involve discrete-time dynamics, projection onto polytopes (e.g., a probability simplex), and equilibria at the boundary of these polytopes. To address these challenges, we discuss extensions of Lyapunov theory, such as those for projected dynamical systems, and their role in understanding convergence to equilibrium. We provide several new results and a coherent framework to study multi-agent systems as dynamical systems based on an underlying game description. We use simple matrix games to illustrate the various concepts. The article combines the literature on online learning, game theory, and dynamical systems in a systematic way and shows the potential but also the limitations of this approach. The analysis is timely because multi-agent systems are widely used but the outcomes are not well understood \citep{julian2019multi}.

Our goal is to equip researchers and practitioners with the theoretical foundations needed to analyze when algorithmic markets are likely to stabilize, and when regulatory intervention may be necessary. The broader question is no longer merely whether an equilibrium exists, but whether and how learning agents will find it.

\paragraph{Organization}
The article is structured as follows.
Section 2 reviews foundational concepts in online learning, including feedback models and regret minimization. Section 3 discusses normal-form games as a central game representation, the Nash equilibrium, and its connection to variational inequalities. 
Section 4 develops the mathematical tools needed to analyze games as dynamical systems, focusing on Lyapunov stability theory for continuous-time unconstrained dynamical systems. Games with continuous actions and utility functions such as the Tullock contest can be analyzed with these tools. However, they find limitations in the analysis of games with discontinuous utility functions such as first-price or all-pay auctions. A workaround avoiding these discontinuities is to discretize these games and analyze the mixed extension of the resulting finite, normal-form game. However, equilibria in such games often lie at the boundaries or vertices of the probability simplex making it a projected dynamical system.
Section 5 discusses discrete-time projected dynamical systems and the difficulties arising in the stability analysis of these systems. We discuss problems in the analysis of such projected dynamical systems and possibilities to study the stability of the Nash equilibria. 
Finally, Section 6 synthesizes the insights and outlines open questions on the stability and regulation of agentic markets.

%%%%%%%%%%%%%%%%%%%%%%%%%%%%%%%%%%%%%%%%%%%%%%%%%%%%%%%%%%%%%%%%%%%%%%%%%%%%%%%%%%%%%%%%%%%%%%%%%%%%%%%%%%
%	2  ONLINE LEARNING					                                                                 %
%%%%%%%%%%%%%%%%%%%%%%%%%%%%%%%%%%%%%%%%%%%%%%%%%%%%%%%%%%%%%%%%%%%%%%%%%%%%%%%%%%%%%%%%%%%%%%%%%%%%%%%%%%
\section{Online Learning}\label{sec:online-learning}

A key challenge for agents in algorithmic markets is setting prices or submitting bids that dynamically respond to market conditions while maximizing expected profit. A fundamental dilemma in designing pricing algorithms is whether to prioritize short-term gains by exploiting known high-yield prices or to explore alternative prices that may yield better long-term outcomes. Online learning algorithms are designed to balance this exploration-exploitation trade-off \citep{bubeck2011introduction}.  

In agentic markets, these algorithms typically operate with \textit{bandit feedback}, meaning that after setting a price, a seller observes only the profit associated with that specific price, without knowing the payoffs of other possible choices. The \textit{multi-armed bandit model} has long been recognized as a natural framework for algorithmic pricing, with early applications dating back to \citet{rothschild1974two}. Today, multi-armed bandit algorithms for pricing are widely studied in the academic literature \citep{trovo2015multi, den2015dynamic, bauer2018optimal, mueller2019low, elreedy2021novel, taywade_multi-armed_2023, qu24}, and practitioners provide extensive resources on implementing these algorithms.  

\subsection{Feedback Models and Algorithms}

Online learning and optimization algorithms are central to pricing strategies on online platforms \citep{mueller2019low, elreedy2021novel, taywade_multi-armed_2023, qu24}. Online optimization concerns sequential decision-making in an uncertain environment, aiming to optimize a performance metric over time.  
A generic model of online learning can be formalized as follows (see Algorithm \ref{alg:online-learning}). At each time step $t = 1,2,\dots$, the agent selects an action $a_t \in \mathcal{A}$ based on a strategy $x_t \in \mathcal{X}$, then receives a payoff $u_t(a_t)$. In algorithmic pricing, this action represents the price set by an agent, and the agent updates its pricing strategy based on observed feedback.  

\begin{algorithm}
\caption{Online Learning}\label{alg:online-learning}
\SetAlgoNlRelativeSize{0} % Normal numbering size
\SetKwInOut{Require}{Require}
\Require{action set $\mathcal{A}$, sequence of payoff functions $u_t: \mathcal{A} \to \mathbb{R}$}
\For{$t = 1, 2, \dots$}{
    select action $a_t \in \mathcal{A}$ according to strategy $x_t$\; 
    realize payoff $u_t(a_t)$ and observe feedback $f_t$\;
    update strategy $x_t \leftarrow x_{t+1}$ using feedback $f_t$\;
}
\end{algorithm}

Agents receive different types of feedback in online learning, including:  
\begin{itemize}
    \item \textit{Bandit feedback}: The agent observes only the payoff of the chosen action.  
    \item \textit{Gradient feedback}: The agent receives some information about how small changes to the action would have influenced the payoff.  
\end{itemize}

\textit{Bandit feedback} is a reasonable model of practical pricing problems. After setting a price in an oligopoly, a seller observes only its realized profit, without knowing how alternative prices would have performed. This is especially true in online marketplaces with many competing sellers and dynamic demand patterns. \textit{Gradient feedback }is a useful abstraction, and gradient-based algorithms turn out to be very effective. Online algorithms are often surprisingly similar even though the feedback mechanism is different, as we will discuss below. 

\subsection{Regret and No-Regret Algorithms}

A key performance measure in online learning is \textit{regret}, which quantifies the loss incurred by an algorithm relative to the best fixed action in hindsight. Formally, the \textit{(external) regret} of an algorithm after $T$ rounds is defined as:  

\begin{equation}
R(T) = \max_{a^* \in \mathcal{A}} \sum_{t=1}^{T} u_t(a^*) - \sum_{t=1}^{T} u_t(a_t)
\end{equation}

where $a^*$ is the optimal fixed action chosen in hindsight, and $a_t$ is the action selected by the algorithm at time $t$. Regret can be analyzed in the adversarial model, where an adversary selects payoffs strategically, potentially reacting to the agent’s past choices and algorithmic strategy.  

An online learning algorithm is called a \textit{no-regret algorithm} if its regret grows sublinearly in $T$, i.e.,  

\begin{equation}
\lim_{T \to \infty} \frac{R(T)}{T} = 0.
\end{equation}

This ensures that, on average, the algorithm performs as well as the best fixed action in the long run.  
As an example, \textit{Exponential Weights (Exp3)} is a well-known no-regret algorithm in the adversarial setting. Exp3 maintains a probability distribution over actions, updating the probabilities based on observed rewards. Actions with higher cumulative rewards are selected with increasing probability, ensuring the algorithm adapts to the underlying payoff structure over time.  

\subsection{Gradient Ascent and Projected Gradient Ascent} 
Gradient ascent is a first-order iterative optimization method used to maximize a function \(u(x)\). Starting from an initial guess \(x_0\), the algorithm updates the iterate by moving in the direction of the gradient:
\begin{equation}
    x_{t+1} = x_t + \eta\,\nabla u(x_t),
\end{equation}
where \(\eta > 0\) is the step size or learning rate. This update rule is based on the observation that \(\nabla u(x_t)\) points in the direction of the steepest ascent, ensuring local improvement in the objective.

In many online learning and pricing applications, however, the decision variable \(x\) is required to lie in a constrained set \(\mathcal{X}\) (e.g., the probability simplex or a compact set of admissible prices). In such cases, after taking a gradient step, the updated point may fall outside \(\mathcal{X}\). To handle this, the \emph{projected gradient ascent} method is used. Its update rule is given by
\begin{equation}
    x_{t+1} = \Pi_{\mathcal{X}}\Bigl(x_t + \eta\,\nabla u(x_t)\Bigr),
\end{equation}
where \(\Pi_{\mathcal{X}}(\cdot)\) denotes the projection onto the feasible set \(\mathcal{X}\). A key property of the projection operator is its non-expansiveness in the Euclidean norm for non-empty, closed, convex sets \(\mathcal{X}\), i.e.,
\[
\|\Pi_{\mathcal{X}}(y) - \Pi_{\mathcal{X}}(z)\| \le \|y-z\|,
\]
which ensures that the projection step does not increase the distance between iterates. This property will be useful in the convergence analysis.

Both methods are widely used in online learning and game theory. In settings with unconstrained decision spaces, gradient ascent can directly be applied. When the feasible set is a nontrivial convex set (as is often the case with probability distributions in pricing or bidding scenarios), projected gradient ascent provides a natural and effective way to incorporate constraints while still following the gradient direction.
The choice of the learning rate \(\eta\) is critical in both methods. If \(\eta\) is too large, the iterates may overshoot or oscillate, while a very small \(\eta\) can lead to slow convergence. 

Projected gradient ascent can be seen as a specific instance within the broader \textit{Follow-The-Regularized-Leader (FTRL)} framework. FTRL is a general approach in online optimization where, at each iteration, the algorithm selects the next action by minimizing the sum of past losses and a regularization term. Exp3 is also an FTRL algorithm. We will focus on the  popular projected gradient ascent algorithm and the resulting game dynamics in this paper, because it is easy to describe and analyze.

%%%%%%%%%%%%%%%%%%%%%%%%%%%%%%%%%%%%%%%%%%%%%%%%%%%%%%%%%%%%%%%%%%%%%%%%%%%%%%%%%%%%%%%%%%%%%%%%%%%%%%%%%%
%	3  GAME THEORETICAL FOUNDATIONS				                                                         %
%%%%%%%%%%%%%%%%%%%%%%%%%%%%%%%%%%%%%%%%%%%%%%%%%%%%%%%%%%%%%%%%%%%%%%%%%%%%%%%%%%%%%%%%%%%%%%%%%%%%%%%%%%

\section{Game Theoretical Foundations}\label{sec:game-theory}

In what follows, we will introduce basics of game-theory, the Nash equilibrium, and variational inequalities, which will be important for the remainder of the article. 

\subsection{Games and the Nash Equilibrium}

We begin with the concepts of a finite normal-form game and Nash equilibria. 
A normal-form game is a representation in game theory that defines the strategies available to each player, their corresponding payoffs, and the resulting outcomes in a simultaneous and strategic interaction. Formally, we have the following definition.

\begin{definition}[(Finite) Normal-form Game\index{normal-form games}]\label{def:normalform}
  A \emph{normal-form game} is a tuple $(\players, \Acal, u)$, where:
\begin{itemize}
    \item $\players = \{1, \dots, \n\}$ is a finite set of players,
    \item $\Acal = \Acal_1 \times \cdots \times \Acal_\n$ is the space of action profiles, where $\Acal_i$ is the action space of player $i$, which may be finite or continuous (e.g., $\Acal_i \subseteq \mathbb{R}^k$ for some $k \geq 1$),
    \item $u = (u_1, \dots, u_\n)$ is a tuple of payoff (or utility) functions, where $u_i: \Acal \to \mathbb{R}$ specifies the payoff of player $i$ for each action profile.
\end{itemize}
If all action spaces $\Acal_i$ are finite, i.e., $|\Acal_i| < \infty$ for all $i \in \players$, the game is called a \emph{finite normal-form game}.
\end{definition}

All agents $i = 1, \dots, \n$ choose their actions $a_i \in \Acal_i$, which together form an action profile $a = (a_1, \dots, a_\n) \in \Acal$. 
We often abbreviate the profile as $(a_i, a_{-i})$, where $a_{-i} = (a_1, \dots, a_{i-1}, a_{i+1}, \dots, a_\n)$ represents the actions of all agents other than $i$. 
Similarly, the combined space of actions for all agents except $i$ is denoted by $\Acal_{-i} = \Acal_1 \times \cdots \times \Acal_{i-1} \times \Acal_{i+1} \times \cdots \times \Acal_\n$.

Agents may also play a distribution over actions, known as a mixed strategy, denoted by $x_i \in \Xcal_i = \Delta(\Acal_i) $, where $\Delta(\Acal_i)$ is the probability simplex over $\Acal_i$. .

\begin{definition}[Mixed Extension] \label{def:mixed_extension}
    Let $ \Gamma = (\Ncal, \Acal, u) $ be a finite normal-form game. In the mixed extension of the game, players choose a strategy $x_i \in \Delta(\Acal_i) \subset \R^{d_i}$. The (expected) utility is then given by
    \begin{equation*}
        u_i(x) = \sum_{a_1 \in \Acal_1} \cdots \sum_{a_n \in \Acal_n} u_i(a_1, \dots, a_n) \cdot (x_{1,a_1} \cdots x_{n,a_n})
    \end{equation*}
\end{definition}

For two-player matrix games, the notation simplifies. Let the vectors $x_1 \in \Delta(\Acal_1) \subseteq \R^m$ and $x_2 \in \Delta(\Acal_2) \subseteq \R^n$ represent the mixed strategy profiles. The expected utilities can be written in matrix-vector form with matrices $A_1, A_2 \in \R^{m \times n}$:
\begin{equation*}
  u_1(x_1, x_2) = x_1^T A_1 x_2, \qquad u_2(x_1, x_2) = x_1^T A_2 x_2
\end{equation*}

A Nash equilibrium is a situation in a strategic interaction where each player's strategy is optimal given the strategies chosen by all other players, and no player has an incentive to unilaterally deviate from their chosen strategy.

\begin{definition}[Nash Equilibrium (NE)]
    In a normal-form game $\nfgdiscr$, a strategy profile $x^* = (x_1^*, \ldots, x_\n^*)$ is a \emph{Nash equilibrium} if, for every player $i \in \Ncal$, we have:
    \begin{equation} \tag{NE}
        u_i(x_i^*, x_\mi^*) \geq u_i(x_i, x_\mi^*), \quad \forall x_i \in \Xcal_i.
    \end{equation}
    If the inequality is strict, we talk about a strict NE.
\end{definition}

%Note that every finite normal-form symmetric game has a Nash equilibrium in mixed strategies \citep{nash1950equilibrium}.  
It is well-known that computing Nash equilibria is PPAD-hard in the worst case \citep{daskalakis2009complexity}. There are alternative solution concepts describing supersets of strategy profiles including the Nash equilibria of a game. Most notably the correlated equilibrium (CE) and the coarse correlated equilibrium (CCE). 
It is well-known that $NE \subseteq CE \subseteq CCE$. While NEs are PPAD-hard to learn in the worst case, CEs and CCEs are not \citep{fosterCalibratedLearningCorrelated1997}. The constraints form a polytope and a (C)CE profile can be selected via linear programming. Importantly, the empirical distribution of play of no-(external)-regret algorithms converges to a CCE. However, CCEs might contain dominated strategies and there might be many strategy profiles that qualify as a CCE. We focus on the NE in this paper.

From now on, we focus on continuous games with continuous action sets and continuously differentiable payoff function \citep{mertikopoulosLearningGamesContinuous2019}. Such games describe a very broad class including mixed extensions of finite, normal-form games but also a Cournot competition or congestion games. Auctions have discontinuities in the utility function. However, when we look at the discretized auction game with a finite number of actions or bids and their mixed extension, it is again a continuous game.   

\begin{definition}[Continuous game]\label{def:gamegradient}
    A continuous game $G=(\Ncal, \Xcal, u)$ is a normal-form game with compact convex action sets $\Xcal_i \subseteq \R^{d_i}$ and continuously differentiable utility functions $u_i: \Xcal \rightarrow \R$ for each player $i$.
\end{definition}

The \emph{game gradient} is defined as the column vector
    \begin{equation}
      F(x) = \left( \frac{\partial u_1(x)}{\partial x_1}, \frac{\partial u_2(x)}{\partial x_2}, \dots, \frac{\partial u_n(x)}{\partial x_n} \right)^T.
    \end{equation}
    
    The \emph{game Jacobian} is the Jacobian matrix of the game gradient, given by
    \begin{equation}
      J(x) = \nabla F(x),
    \end{equation}
    where the $(i, j)$-th entry (or block) of $J(x)$ is
    \begin{equation}
      J_{ij}(x) = \frac{\partial}{\partial x_j} \left( \frac{\partial u_i(x)}{\partial x_i} \right) = \frac{\partial^2 u_i(x)}{\partial x_j \partial x_i}.
    \end{equation}

The diagonal entries correspond to the Hessian of $ u_i $ w.r.t. $ x_i $. In all other entries, we have the Jacobian w.r.t. $ x_j $ of the gradients $ \nabla_{x_i} u_i $. The Jacobian matrix often has a block structure since the $x_i$ can be multidimensional, in particular in mixed-extensions of finite, normal-form games.

If the Jacobian has a symmetric structure, then it is a \textit{potential game}. This is because agents' payoffs align with a common potential function. Several papers analyzed the convergence of learning algorithms in such potential games \citep{mondererPotentialGames1996a}. Another condition for which we know that learning algorithms converge is that of strategic complements \citep{milgrom1991adaptive}. Positive cross-derivatives are a sufficient condition for \textit{strategic complements}: 
\begin{equation*}
  \frac{\partial^2 u_i(x)}{\partial x_i \partial x_j} \geq 0, \quad \forall i \neq j.   
\end{equation*}
This ensures that an increase in player $j$'s strategy $x_j$ increases the marginal utility of player $i$'s strategy $x_i$.

\subsection{Variational Inequalities and Nash Equilibria}

For continuous games, one can write down the first-order conditions for a Nash equilibrium.
Aggregating these conditions in the game gradient over all players, we have, for any \( x = (x_1, \dots, x_n) \in \mathcal{X} \),
\begin{equation*}
  \sum_{i=1}^{n} \langle \nabla_{x_i} u_i(x_i^*, x_{-i}^*),\, x_i - x_i^* \rangle \le 0,
\end{equation*}
or, equivalently, \(\langle F(x^*),\, x - x^* \rangle \le 0 \) for all \( x \in \mathcal{X}\), which is a variational inequality. 

\begin{definition}
  Consider the function $ F: \Xcal \rightarrow \R^n $ where $ \Xcal \subseteq \R^n $ is a non-empty closed set. A (finite) \textit{variational inequality} $ VI(\Xcal,F) $ is the problem of finding a vector $ x^* \in \Xcal $ such that
  \begin{equation}\label{eq:vi} \tag{VI}
    \langle F(x^*), x-x^* \rangle \leq 0 \quad \forall x \in \Xcal
  \end{equation}
  We call $ x^* $ the solution and $ \Xcal $ the feasible set of $ VI(\Xcal,F) $.
\end{definition}

Often, the set $ \Xcal $ is convex, which will be assumed in the following. For example, $x_i \in \Delta(\Acal_i) $, where $\Delta(\Acal_i) = \Xcal_i$ denotes the probability simplex over $\Acal_i$ in a finite, normal-form game.

To connect the notion of Nash equilibria with variational inequalities, we have to make the assumption that the underlying utility functions of the game are (individually) concave. 

\begin{definition}
	A game is called \textit{(pseudo-)concave} if it satisfies the individual concavity assumption, that is to say that
	\begin{equation*}
		u_i(x_i,x_{-i}) \text{ is (pseudo-)concave in } x_i \quad \forall x_{-i} \in \Xcal_{-i}, \, i \in \Ncal
	\end{equation*}
\end{definition}

\begin{lemma}[\citet{mertikopoulosLearningGamesContinuous2019}]
    If $x^* \in \Xcal$ is a Nash equilibrium, it is a solution to $VI(\Xcal, F)$. 
    If the game is (pseudo-)concave, any solution to $VI(\Xcal, F)$ is also a Nash equilibrium.
\end{lemma}

The mixed extension of any normal-form game introduced above is concave, i.e. the individual utility functions are concave. In fact, they are even linear in mixed extensions. We assume individual concavity of the utility functions from now on, although this assumption needs to be checked for every game.  

For interior points, we can characterize the problem in a different way. The following result connects variational inequalities to the first-order conditions in unconstrained optimization in such cases.

\begin{theorem}
    Let $\mathcal{X} \subset \mathbb{R}^n$ be a convex set with nonempty interior, and let $F: \Xcal \to \mathbb{R}^n$ be a mapping. 
    Suppose that $x^* \in \operatorname{int}(\mathcal{X})$. 
    If $x^*$ is a solution of the variational inequality VI$(\mathcal{X},F)$, then $F(x^*) = 0$.
\end{theorem}
The proof can be found in any textbook on variational inequalities.
\begin{proof}
    Suppose first that $x^* \in \operatorname{int}(\mathcal{X})$ is a solution of $\text{VI}(\mathcal{X},F)$, i.e.,
    \( \langle F(x^*), x - x^* \rangle \leq 0 \, \forall\, x \in \mathcal{X}\).
    Since $x^*$ is an interior point of $\mathcal{X}$, there exists some $t > 0$ such that \(x = x^* + t \, F(x^*) \in \Xcal \). 
    Substitute $x = x^* + t\,F(x^*)$ into the variational inequality to obtain
    \[
    \langle F(x^*), (x^* + t\,F(x^*)) - x^* \rangle =  t\,\langle F(x^*), F(x^*) \rangle = t \Vert F(x^*) \Vert^2 \geq 0
    \]
    Combined with the variational inequality, we find that $F(x^*) = 0$.
\end{proof}

\subsection{Existence and Uniqueness of the Nash Equilibrium}\label{sec:unique}

Now that we have introduced the Nash equilibrium, we discuss existence and uniqueness. We reason that, under mild assumptions, there must be at least one Nash equilibrium. In order to show this, we draw on the connection between Nash equilibria and solutions of the variational inequality $VI(\Xcal, F)$ established in the previous section. We relate the solutions of $VI(\Xcal, F)$ with the equilibrium points of a specific mapping which allows us to apply Brouwer's fixed point theorem.

\begin{theorem}[Brouwer's Fixed Point Theorem]
    {Let $\Xcal \subseteq \mathbb{R}^n$ be a non-empty, compact, and convex set. If $g: \Xcal \to \Xcal$ is a continuous function, then there exists a point $x^* \in \Xcal$ such that: $g(x^*) = g^*$.}
\end{theorem}

To use this result, we show that $ VI(\Xcal,F) $ is equivalent to a fixed point problem. 

\begin{lemma}[\citet{geigerVariationsungleichungen2002}] \label{thm:vi_fp}
    Consider the problem $ VI(\Xcal,F) $ where $ \Xcal $ is a non-empty, closed, and convex set. The point $ x^* \in \Xcal $ is a solution of $ VI(\Xcal,F) $ if and only if $ x^* $ is a fixed point of the mapping $ P(x) := \proj_{\Xcal}(x-\gamma F(x)) $, i.e., $ x^* = P(x^*) $ for some $ \gamma > 0 $.
\end{lemma}

The existence follows now directly from the previous theorem, where $ P(x) $ is the mapping $ f $, and $\proj_X$ is a projection operator. Continuity follows from continuity of $ F $ and $\proj_{\Xcal}$. If we are interested in \textit{uniqueness}, we can make an additional assumption on $ F $.

\begin{lemma}[\citet{dupuis1993dynamical}]
    Let $ \Xcal \subseteq \R^n $ be a closed convex set and $ F: \Xcal \rightarrow \R^n $ strictly monotone. Then $ VI(\Xcal,F) $ has exactly one solution.
\end{lemma}

Strict monotonicity is central for this and is defined as follows.

\begin{definition}\label{def:monotonicity}
    Let $ \Xcal \subset \R^n $. A function $ F: \Xcal \rightarrow \R^n $ is 
    \begin{enumerate}[label=\roman*)]
        \item \textit{monotone} if
        \begin{equation*}
            \langle F(x) - F(y), x-y \rangle \leq 0 \quad \forall x,y \in \Xcal
        \end{equation*}
        \item \textit{strictly monotone} if
        \begin{equation*}
             \langle F(x) - F(y), x-y \rangle  < 0 \quad \forall x \ne y \in \Xcal
        \end{equation*}
        \item \textit{strongly monotone} if there exists an $\alpha > 0$ such that
        \begin{equation*}
            \langle F(x) - F(y), x-y \rangle \leq - \alpha \Vert x - y \Vert^2  \quad \forall x,y \in \Xcal
        \end{equation*}
        \item \textit{pseudomonotone} if 
            \begin{equation*}
            \langle F(y), x-y \rangle \leq 0 \quad \Rightarrow \quad \langle F(x), x-y \rangle\leq 0, \quad \forall x,y \in \Xcal.
            \end{equation*}
        \end{enumerate}
\end{definition}

The monotonicity conditions in  \citet{mertikopoulosLearningGamesContinuous2019} or \citet{rosenExistenceUniquenessEquilibrium1965} correspond to \textit{strict} monotonicity. 
To verify monotonicity we can use different characterizations using the Jacobian $ J(x) = \nabla F(x)$ of $ F $.

\begin{theorem}[\citep{geigerVariationsungleichungen2002}] \label{thm:mon_jac}
    Let $ \Xcal \subseteq \R^n $ be an open and convex set and assume that $ F: \Xcal \rightarrow \R^n $ is continuously differentiable. Then we have
    \begin{enumerate}[label=\roman*)]
        \item $ F $ is monotone (on $ \Xcal $) if and only if $ J(x) $ is negative semi-definite for all $ x \in \Xcal $. 
        \item If $ J(x) $ is negative definite for all $ x \in \Xcal $, then $ F $ is strictly monotone on $ \Xcal $. 
    \end{enumerate}
\end{theorem}

\subsection{Variational Stability and Strict Nash Equilibrium} \label{sec:variational_stability}

As we discussed earlier, if the individual utility functions are concave and $ x^* $ solves the variational inequality $ VI(\Xcal,F) $ (see equation \ref{eq:vi}), these solutions are equivalent to the Nash equilibria of the corresponding game. Apart from this Stampaccia-type VI, we can also consider the Minty-type variational inequality, which is sometimes referred to as weak solution or the dual VI \citep{geigerVariationsungleichungen2002}.

	\begin{definition}[Minty-type VI (MVI)]
		The point $ x^* \in \Xcal $ with $\Xcal$ being convex and compact is a \textit{weak} solution if it satisfies
		\begin{equation} \tag{MVI}
			\langle F(x), x-x^* \rangle \leq 0 \quad \forall x \in \Xcal
		\end{equation}
	\end{definition}

It is known that if $F$ is continuous on $\Xcal$, then each solution to the MVI is a solution to a Stampaccia-type VI \citep{cavazzutiNashEquilibriaVariational2002}. A strict version of this definition is referred to as variational stability \cite{mertikopoulosLearningGamesContinuous2019}:

 \begin{definition}
     A point $ x^* \in \Xcal $  with $\Xcal$ being convex and compact is variationally stable if there exists a neighborhood $ U $ of $ x^* $ such that
     \begin{equation}\label{eq:vs} \tag{VS}
         \langle F(x), x-x^* \rangle \leq 0 \quad \forall x \in U
     \end{equation}
 with equality if and only if $ x^*=x $.
 \end{definition} 

A number of results about variationally stable points in games were shown by \citet{mertikopoulosLearningGamesContinuous2019}: 
  \begin{enumerate}[label=\roman*)]
      \item If the game is (pseudo-)concave and $ x^* $ is variationally stable, then $ x^* $ is an isolated Nash equilibrium of $ G $. 
      \item In finite games: $ x^* $ being variationally stable is equivalent to $ x^* $ being a strict Nash equilibrium. 
      \item If $ x^* $ is globally variationally stable, it is the game's unique Nash equilibrium. 
%      \item In continuous games, sharp equilibria are variational stable. 
  \end{enumerate}

Importantly, \citet{mertikopoulosLearningGamesContinuous2019} show that variational stability guarantees that the induced sequence of no-regret learning converges globally to globally variationally stable equilibria, and it converges locally to locally variationally stable equilibria.
\citet{giannou2021survival} analyze no-regret learning with various types of feedback during the learning (e.g., gradient or bandit feedback). They establish that a Nash equilibrium is stable and attracting with arbitrarily high probability \textit{if and only if} it is strict, i.e. it is variationally stable. This is an important insight that we will leverage below.

Global variational stability follows from strict monotonicity, can be shown directly using the definition, or specifically for finite normal-form games by considering all pure deviations from the pure equilibrium strategy:

\begin{proposition}
\label{prop:variational-stability-by-pure-deviations}
	A PNE $a^* \in \Acal$ (in mixed notation: $x^* \in \Delta(\Acal)$) of a finite game in mixed-extension is globally variationally stable if and only if, for any \emph{pure} deviation $x \in \Delta(\Acal)$
	\begin{equation*}
		\sum_{i = 1}^{n} \left(u_i(a) - u_i(a_i^*, a_{-i})\right) < 0 \quad \forall a \neq a^*.
	\end{equation*}
\end{proposition}

The proof can be found in the appendix.

Strictly monotonous games are globally variationally stable \citep{mertikopoulosLearningGamesContinuous2019}, and certain gradient dynamics converge to the unique equilibrium of such games. Variational stability is a weaker condition that still implies convergence to the unique Nash equilibrium under certain dynamics if it holds globally \citep{mertikopoulos2016learning, mertikopoulosLearningGamesContinuous2019}.

%%%%%%%%%%%%%%%%%%%%%%%%%%%%%%%%%%%%%%%%%%%%%%%%%%%%%%%%%%%%%%%%%%%%%%%%%%%%%%%%%%%%%%%%%%%%%%%%%%%%%%%%%%
%	4  GAMES AS DYNAMICAL SYSTEMS		                                                    		     %
%%%%%%%%%%%%%%%%%%%%%%%%%%%%%%%%%%%%%%%%%%%%%%%%%%%%%%%%%%%%%%%%%%%%%%%%%%%%%%%%%%%%%%%%%%%%%%%%%%%%%%%%%%

\section{Games as Dynamical Systems}\label{sec:unconstrained-dynamics}
A Nash equilibrium is commonly considered to be the outcome of the interaction between rational agents.
However, it is acknowledged that the analytical solution requires extensive information about the game (e.g. payoffs of all players) and that there must be some sort of coordination between the agents in the case where there are multiple Nash equilibria. Very much like a variational inequality, the Nash equilibrium describes a static concept, but it says little how agents find equilibrium. 

If learning agents that selfishly seek to maximize their payoff \textit{without prior information about other's values or costs} reach a Nash equilibrium, this is a strong rationale for the Nash equilibrium as a prediction. A principled way to study the learning process of independent agents is via dynamical systems.

A \textit{dynamical system} is a mathematical framework used to describe the evolution of a state variable $x(t) \in \mathbb{R}^n$ over time $t$. In this section, we focus on continuous-time dynamics. 
In continuous time, we will consider autonomous dynamical systems represented by ordinary differential equations (ODEs) of the form
\begin{equation}
  \dot{x}(t) = f(x(t)),
\end{equation}
where $f: \mathbb{R}^n \to \mathbb{R}^n$ is a continuous function, and $\dot{x}(t)$ denotes the time derivative of $x(t)$. 
The learning algorithms that we introduced in Section \ref{sec:online-learning} induce discrete-time dynamics. However, with sufficiently small step sizes, they approximate the continuous dynamics, which can then serve as a model. A dynamical system in which the rules governing its evolution do not depend explicitly on time is also called an \textit{autonomous system}. Here, the system is only determined by its current state.

Gradient-based learning algorithms in games can be described as a dynamical system where the state variables $x$ are a combination of all strategies and the function $f$ is associated with the game gradient $F$ (see Definition \ref{def:gamegradient}): 
\begin{equation}
  \dot{x}(t) = F(x(t)).
\end{equation}
In this chapter of our survey, we will first focus on unconstrained dynamics where no modifications to the game gradient are required. This is useful for many continuous games. If the state space $\Xcal$ of the dynamical system is constrained, we will need projections to modify our dynamics accordingly, e.g., $f(x) = \proj(F(x))$ for some projection $\proj$. We will cover this in Section \ref{sec:projected-dynamics}. 

In the following, we will introduce concepts for the analysis of stability in dynamical systems, starting with local stability. We will mostly rely on the canonical continuous-time view, but we note that many of the following concepts and theorems also apply to the discrete-time case. 

\subsection{Local Stability} 

Stability captures the scenario that a trajectory does not move far away from a given equilibrium point. This strengthens the meaning of an equilibrium point. 
One way to characterize stability is the notion of \textit{Lyapuov stability}.
We denote the flow (also: “trajectory”, “orbit”) of the dynamical system $\dot{x} = f(x)$, starting at $x(0) = x_0$, by $\varphi_t(x_0)$.

\begin{definition}[Lyapunov stability]
    % \label{def:stability}
    An equilibrium point $x^*$ of a flow $\varphi_t$ is (Lyapunov) \emph{stable} if, for each neighborhood $\Ucal$ of $x^*$, there exists a (potentially smaller) neighborhood $\Vcal$ such that
    \begin{equation}
        \forall x \in \Vcal, \forall t \geq 0: \quad \varphi_t(x) \in \Ucal % TODO: Notation \varphi
    \end{equation}
\end{definition}

\begin{figure}[h]
    \centering
  %  \textbf{Lyapunov stability}
   % \vspace{1em}
  
        \centering
    \begin{tikzpicture}[scale=0.7]
        
        % Equilibrium point x^*
        \filldraw[black] (0,0) circle (2pt) node[above right] {\footnotesize $x^*$};
        
        % Outer region U
        \begin{scope}
            \clip plot[smooth cycle, tension=0.7] coordinates { (3,0) (2.6,1.8) (1.5,2.7) (0,3) (-1.5,2.7) (-2.6,1.8) (-3,0) (-2.6,-1.8) (-1.5,-2.7) (0,-3) (1.5,-2.7) (2.6,-1.8)};
            \fill[niceBlue!30, opacity=0.1] (-3,-3) rectangle (3,3);
        \end{scope}
        \draw[dashed, thick,niceBlue] plot[smooth cycle, tension=0.7] coordinates { (3,0) (2.6,1.8) (1.5,2.7) (0,3) (-1.5,2.7) (-2.6,1.8) (-3,0) (-2.6,-1.8) (-1.5,-2.7) (0,-3) (1.5,-2.7) (2.6,-1.8)};
        \node[niceBlue] at (2.5, 2.5){\footnotesize $\Ucal$};
        
        % Inner region V
        \begin{scope}
            \clip plot[smooth cycle, tension=0.8] coordinates { (1.5,0) (1.3,0.9) (0.8,1.4) (0,1.5) (-0.8,1.4) (-1.3,0.9) (-1.5,0) (-1.3,-0.9) (-0.8,-1.4) (0,-1.5) (0.8,-1.4) (1.3,-0.9)};
            \fill[red!30, opacity=0.1] (-1.5,-1.5) rectangle (1.5,1.5);
        \end{scope}
        \draw[dashed, thick, niceRed] plot[smooth cycle, tension=0.8] coordinates { (1.5,0) (1.3,0.9) (0.8,1.4) (0,1.5) (-0.8,1.4) (-1.3,0.9) (-1.5,0) (-1.3,-0.9) (-0.8,-1.4) (0,-1.5) (0.8,-1.4) (1.3,-0.9)};
        \node[niceRed] at (1.5, 1.5){\footnotesize $\Vcal$};
        
        % Trajectory
        \filldraw[black] (0.5, 0.7) circle (2pt) node[above left] {\footnotesize $x$};
        \draw[thick, ->, smooth] plot coordinates {(0.5, 0.7) (1,1) (2, -1) (0, -1) (-0.5, 1)}
        node[right] {\footnotesize $\varphi$};

    \end{tikzpicture}
    \caption{Visualization of Lyapunov stability.}
    \label{fig:lyap_stability}
\end{figure}
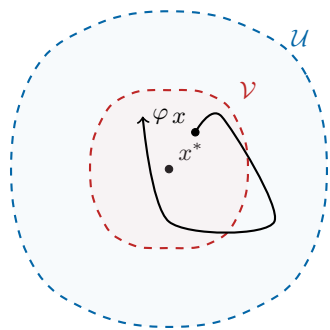
    
\begin{definition}[Lyapunov function \citep{meiss_differential_2017}]
    \label{def:lyapunov-function}
    A continuous function $V: \mathbb{R}^n \to \mathbb{R}$ is a (strong) Lyapunov function for an equilibrium $x^*$ if, on an open neighborhood $\Ucal$ of $x^*$,
    
    \begin{enumerate}[label=\roman*)]
        \item $V(x^*)=0$ and $V(x)>0$, and \label{lyapunov:positive-definiteness}
        \item $V(\varphi_t(x)) < V(x)$ for all $t > 0$. \label{lyapunov:decrescence}
    \end{enumerate}

    It is a weak Lyapunov function if the inequality in condition \ref{lyapunov:decrescence} is not strict. A sufficient condition for condition \ref{lyapunov:decrescence} is that the Lie derivative 
    \begin{equation*}
        \frac{d V}{d t} = \nabla V(x) \cdot f(x)
    \end{equation*}
    is negative.
\end{definition}

In a similar manner, Lyapunov functions can also be defined for discrete-time systems. In this case, the second condition is replaced by the condition that the Lyapunov function does not increase along trajectories of the dynamics, $V(x^{t+1}) - V(x^t) \leq 0$, with equality only for $x^t = x^{t+1}$.
In addition, asymptotic stability can be considered:
\begin{definition}[Asymptotic stability]
An equilibrium point \({x}^*\) is \textit{asymptotically stable} if:
\begin{enumerate}[label=\roman*)]
    \item It is Lyapunov stable.
    \item Trajectories starting sufficiently close to the equilibrium satisfy:
        \[
            \|{x}(t=0) - {x}^*\| < \delta \Rightarrow \lim_{t \to \infty} {x}(t) = {x}^*.
        \]
\end{enumerate}
\end{definition}
This means that trajectories not only remain close but also converge to the equilibrium as \(t \to \infty\).
Asymptotic stability implies Lyapunov stability, but not vice versa. A Lyapunov function implies asymptotic stability if the function value decreases strictly over time ($\dot{V}(x) < 0$ for all $x \in U \setminus \{x^*\}$) (a strong Lyapunov function).
Exponential stability is an even stronger condition that ensures convergence occurs at an exponential rate. 

\citet{meiss_differential_2017} provides the corresponding stability result:

\begin{theorem}[Lyapunov stability]
    Let $x^*$ be an equilibrium point of the flow $\varphi_t(x)$. If $V$ is a weak Lyapunov function, $x^*$ is stable. If $V$ is a strong Lyapunov function, then $x^*$ is asymptotically stable.
\end{theorem}

Finding a Lyapunov function can be a challenging task. For unconstrained dynamical systems, we can rely on linearizations of the dynamical system to show local stability. The local dynamics of a game near an equilibrium can be analyzed via approximation of the (continuously differentiable) game gradient $F$. Let 
\begin{equation*}
    \begin{split}
	F(x) &= F(x_0) + \nabla_x F(x_0) (x - x_0) + \mathcal{O}((x - x_0)^2) \\
    &= F(x_0) + J(x_0) (x - x_0) + \mathcal{O}((x - x_0)^2)
    \end{split}
\end{equation*}
be the multivariate Taylor expansion around $x_0$.
In a small neighborhood $U$, the dynamics can be approximated as
\begin{equation*}
	F(x) \approx F(x_0) + J(x_0) (x - x_0)
\end{equation*}
by dropping higher-order terms. Suppose $x_0$ is an interior equilibrium of $F$. Then $F(x_0)~=~0$ and $F(x) \approx J(x_0) (x - x_0)$. Thus
\begin{equation*}
	\frac{d}{dt}{(x - x_0)} = \dot{x} - \dot{x_0} = \dot{x} = J(x_0) (x - x_0),
\end{equation*}
so the variable $\Delta x = (x - x_0)$ evolves (approximately) according to the linear system $\dot{\Delta x} = J(x_0) \Delta x$. From linear system analysis, we get asymptotic stability of $\Delta x = 0$ (and thus $x_0$) if the eigenvalues of the Jacobian matrix $J(x)=\nabla F(x_0) \in \mathbb{R}^{n \times n}$ are all negative (see Theorem \ref{thm:mon_jac}). Note that eigenvalues can be complex and only the real part has to be negative.

\subsection{Global Asymptotic Stability and Basin of Attraction}\label{sec:globalstability}

Apart from local asymptotic stability around an equilibrium, we are interested in the region of convergence. 
Such conditions can hold in a specific region or globally. We will start with the latter case, \textit{global asymptotic stability}.

\begin{definition}[Global asymptotic stability (GAS)]
A point \( x^* \) is globally asymptotically stable (GAS) if:
\begin{enumerate}[label=\roman*)]
\item \( x^* \) is \textit{globally stable}, i.e., for all \( \varepsilon > 0 \), there exists \( \delta > 0 \) such that if \( \|x(0) - x^*\| < \delta \), then \( \|x(t) - x^*\| < \varepsilon \) for all \( t \geq 0 \).
\item \( x^* \) is \textit{globally attractive}, i.e., for all initial conditions \( x(0) \in \mathbb{R}^n: \lim_{t \to \infty} x(t) = x^* \).
\end{enumerate}
\end{definition}

Global asymptotic stability can be shown by a Lyapunov function $V$ that is globally positive definite and whose derivative is globally negative definite. "Globally" essentially means that the neighborhood becomes the entire set: $U = \Xcal$. Also, $V$ must be radially unbounded, i.e., \( \lim_{\norm{x} \to \infty} V(x) = \infty \).
For a linear system \( \dot{x} = Ax \), global asymptotic stability is ensured if the matrix \( A \) is Hurwitz, i.e., all eigenvalues of \( A \) have strictly negative real parts.

If a point is locally asymptotically stable but not globally, we seek to characterize its \textit{basin of attraction}. The basin of attraction for a given equilibrium $x^*$ is defined as the set of all initial states for which the dynamics converge to $x^*$.
While it is, in general, difficult to specify the entire basin of attraction, we can at least determine subsets of it by means of the following result.\footnote{For a more thorough introduction, see, e.g., \cite{khalil_nonlinear_2002}.}

\begin{corollary}[Estimating the basin of attraction] \label{cor:basin_attraction}
    For an equilibrium $x^*$ that is asymptotically stable according to the Lyapunov function $V$, let 
    \begin{equation*}
        \Vcal = \{x^*\} \cup \{x \vert V(x) > 0, \dot{V} < 0\} \quad \text{and} \quad \Ucal_c = \{x \vert V(x) \leq c\}.
    \end{equation*}
    If $\Ucal_c \subseteq \Vcal$ and $\Ucal_c$ is bounded, $\Ucal_c$ is a subset of the basin of attraction.
\end{corollary}

If $c$ is too big, then $\Ucal_c$ may contain points where the dissipation property of the Lyapunov function, i.e., $\dot{V}<0$ outside equilibrium, fails. If this is not the case, then $\Ucal_c$ is a subset of the basin of attraction.

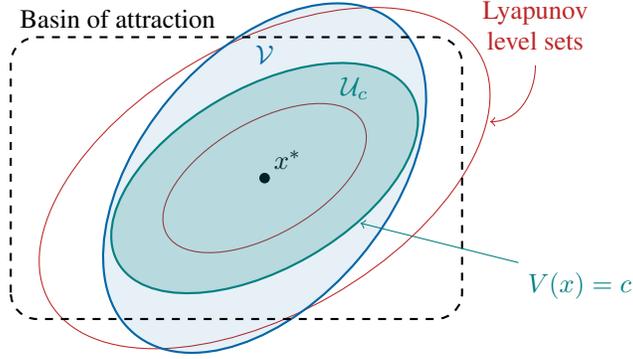
\begin{figure}[h]
    \centering
    \begin{tikzpicture}[x=0.75cm,y=0.75cm]
        
        % Equilibrium point
        \fill[black] (0,0) circle (2pt) node[above right] {$x^*$};
        
        % Lyapunov function isolines (diagonal ellipses)
        \draw[niceRed, rotate around={30:(0,0)}] (0,0) ellipse (4.4 and 2.4);
        \draw[niceRed, rotate around={30:(0,0)}] (0,0) ellipse (2 and 1);
        \node[niceRed, align=center] at (4.8,2.7) {Lyapunov \\ level sets};
        \draw[->, niceRed, rounded corners] (4.8, 2.0) to[out=-90, in=0] (4.0, 1.0);
        
        % Set V
        \draw[niceBlue, thick, rotate around={-40:(0,0)}, fill=niceBlue, fill opacity=0.1] (0,0) ellipse (2.2 and 3.6);
        \node[niceBlue] at (0.0, 2.2) {$\Vcal$};
        
        % Set U_c
        \draw[teal, thick, rotate around={30:(0,0)}, fill=teal, fill opacity=0.2] (0,0) ellipse (3 and 1.6);
        \node[teal] at (1.6, 1.6) {$\Ucal_c$};
        
        % Basin of attraction
        \draw[draw=black, thick, dashed, rounded corners=10] (-4.5, -2.5) rectangle (3.5, 2.5);
        \node[black, anchor=south west] at (-4.5, 2.5) {Basin of attraction};
        
        % Arrow and label for the second level set
        \draw[<-, teal] (1.7, -0.8) -- (4.5, -1.5) node[below right, teal] {$V(x) = c$};
    \end{tikzpicture}
    \caption{Estimating the basin of attraction with Corollary \ref{cor:basin_attraction}.}
    \label{fig:basin-of-attraction}
\end{figure}

This result follows immediately from LaSalle's invariance principle \citep{lasalle_extensions_1960}, which we state below for completeness.
\begin{theorem}[LaSalle's Invariance Principle]
    Let $\Ical \subseteq \Xcal$ be a  compact \emph{invariant set} under the dynamics $f$, i.e., 
    \begin{equation*}
      x(0) \in \Ical \Rightarrow x(t) \in \Ical, \quad \forall t \geq 0.
    \end{equation*}
   Let $V$ be a continuously differentiable function with $\dot{V}(x) \leq 0$ on $\Ical$, and define $\mathcal{E} := \{x \in \Ical \vert \dot{V}(x) = 0\}$. 
   Then, each trajectory of the dynamical system approaches $\mathcal{M}$, the largest invariance set in $\mathcal{E}$.
\end{theorem} 

By selecting $c$ appropriately, we can ensure that the only place where $\dot{V}(x)$ can be zero is at the equilibrium point $x^*$.
In summary, by using Lyapunov functions together with LaSalle's invariance principle, you can derive stability results that extend beyond just an infinitesimally small neighborhood of an equilibrium (i.e., local stability) to larger, well-defined regions, called regions of attraction. 

While asymptotic stability or Lyapunov stability describe the quality of convergence, an \textit{attractor} represents the set where trajectories end up. In other words, an \textit{attractor} is a set (often a point or a more complex structure) in a dynamical system such that trajectories starting from a certain neighborhood (the basin of attraction) asymptotically approach the attractor. While equilibrium points can be attractors, attractors are a broader concept in dynamical systems and can take other forms, such as periodic orbits or limit cycles, chaotic sets, or more complex structures. Chain recurrent sets capture all the types of recurrent behaviors—in particular, all the attractors (point equilibria, limit cycles, chaotic attractors, etc.) are contained within the chain recurrent set.  Attractors are asymptotically stable (or at least Lyapunov stable) in their basins, while chain recurrent classes need not be stable in any attracting sense. They can include, for example, saddle points.

Point attractors are equilibrium points, and this is what we will focus on. Trajectories may oscillate, spiral, or otherwise behave non-monotonically while converging to the attractor.  
A \textit{monotone attractor} is a special type of attractor where the distance between the trajectory and the attractor (e.g., an equilibrium point) is non-increasing over time.
More specifically, an equilibrium point \(x^* \in \mathcal{X} \subseteq \mathbb{R}^n\) is a \emph{monotone attractor} if there exists a neighborhood \(B(x^*, \delta) \subseteq \mathcal{X}\) such that for any trajectory \(x(t)\) with \(x(0) \in B(x^*, \delta)\), the Euclidean distance \(\|x(t) - x^*\|\) is a non-increasing function of time \(t \geq 0\):
\begin{equation*}
  \frac{d}{dt} \|x(t) - x^*\| \leq 0, \quad \forall t \geq 0.
\end{equation*}

If the distance \(\|x(t) - x^*\|\) is strictly decreasing (unless \(x(t) = x^*\)), then \(x^*\) is called a \emph{strictly monotone attractor}.

An equilibrium point \(x^* \in \mathcal{X}\) is a \textit{global monotone attractor} if the monotonicity property holds for all trajectories starting from any initial condition \(x(0) \in \mathcal{X}\). Specifically, the Euclidean distance \(\|x(t) - x^*\|\) satisfies:
\begin{equation*}
  \frac{d}{dt} \|x(t) - x^*\| \leq 0, \quad \forall t \geq 0, \, \forall x(0) \in \mathcal{X}.
\end{equation*}
If the distance is strictly decreasing globally, \(x^*\) is a \textit{strictly global monotone attractor}.

Every strictly asymptotically stable equilibrium is a strictly monotone attractor, but not every attractor is  asymptotically stable. A weak attractor may attract nearby trajectories, but the convergence may not be strict or asymptotically fast.
\citet{flokas2020no} show that mixed Nash equilibria cannot be stable and attracting under certain no-regret dynamics. 
For example, consider a continuous Matching Pennies game where each player chooses a strategy in the interval \([0,1]\). The payoff functions are defined as
\begin{equation*}
  u_1(x,y) = \sin\bigl(2\pi (y - x)\bigr), \quad u_2(x,y) = -\sin\bigl(2\pi (y - x)\bigr),
\end{equation*}
for all \(x,y \in [0,1]\). In this game, no pure strategy Nash equilibrium exists, and only a mixed Nash equilibrium is present, where each player randomizes uniformly over the interval $[0,1]$.

\subsection{Application to Continuous Games with an Interior Equilibrium}\label{sec:vs}
In the following, we will look at two specific games and use the introduced tools to analyze the resulting dynamical system of the continuous games and in particular their equilibria which lie in the interior of the action sets.

\subsubsection{Tullock Contest}

The Tullock Contest can be seen as a smooth approximation of the all-pay auction with differentiable utility functions for positive actions.
The utility function for player \( i \) in the Tullock contest is
\begin{equation*}
  u_i(x_i, x_{-i}) = \frac{x_i^r}{x_i^r + x_{-i}^r} V - x_i,
\end{equation*}
where  \( x_i \geq 0 \) is the bid of player \( i \), \( V > 0\) the value of the prize, and \( r > 0 \) a parameter of the contest success function. 
In the following analysis, we focus on the 2-player setting with \(V=1\) and \(r=2\). 
Usually, the utility function at $a = 0$ is given by $u_i(x_i,x_{-i}) = \tfrac{1}{n}$. 
To avoid this discontinuity in our analysis, we additionally assume that the feasible set $\Xcal_i$ only contains strict positive \( x_i \).

\paragraph{Nash Equilibrium} Since $\tfrac{\partial^2}{\partial x_i^2} u_i(x_i,x_{-i}) \leq 0$ the game is concave and every solution $a^*$ of the variational inequality is also a Nash equilibrium. The variational inequality is given by:
\[ \langle F(x^*), x-x^* \rangle = \left(\dfrac{2x_1 x_2^2}{(x_1^2+x_2^2)^2} - 1\right)(x_1 - x_1^*) + \left( \dfrac{2x_1^2 x_2}{(x_1^2+x_2^2)^2} - 1 \right) (x_2-x_2^*) \leq 0, \quad \forall x \in \Xcal. \]
Assuming that our equilibrium lies in the interior of our feasible set $\Xcal$,  \(F(x^*)=0\) has to hold. 
Solving for a symmetric solution, we can find the symmetric equilibrium \( x^*=(x_1^*, x_2^*) = (\tfrac 1 2, \tfrac 1 2)\).

\paragraph{Local Stability} Using the linear approximation of the system around the equilibrium $x^*$, we can analyze local stability by checking the eigenvalues of the Jacobian matrix $J(x^*) = \nabla F(x^*)$. In our example, we have 
\[J(x^*) = \begin{bmatrix}
-\frac{1}{2} & 0 \\
0 & -\frac{1}{2}
\end{bmatrix}.
\] 
The eigenvalues are negative, which implies local asymptotic stability. 

\paragraph{Variational Stability} To get global convergence results for some gradient-based learners, we would need global variational stability of the equilibrium (see Section~\ref{sec:variational_stability}), i.e., 
\[ \langle F(x), x-x^* \rangle = \left( \frac{2x_1 x_2^2}{(x_1^2 + x_2^2)^2} - 1 \right)  (x_1 - \tfrac 1 2)  + 
\left( \frac{2x_2 x_1^2}{(x_1^2 + x_2^2)^2} - 1 \right) (x_2 - \tfrac 1 2) < 0, \quad \forall x \in \Xcal\setminus\{x^*\}\]
Unfortunately, there are points for which the condition is violated (see Figure~\ref{fig:minty}). 
But we can get local variational stability for the set $ \Bcal_c(x^*) =\{ x \in \Xcal: \Vert x - x^* \Vert \leq c \}$ with $c=\tfrac 1 3$
Interestingly, (local) variational stability is equivalent to having a specific (local) Lyapunov function:

\begin{remark} \label{rem:vs_lyap}
    The negative definiteness of the derivative $\dot V(x)$ with $V(x) = \tfrac 1 2 \Vert x - x^* \Vert_2^2$ on some neighborhood $U$ of $x^*$ is equivalent to $x^*$ being variationally stable on $U$:
    \begin{equation*}
      0 > \dot V(x) = \nabla V(x)^T F(x) = \langle F(x), \nabla \tfrac 1 2 \Vert x - x^* \Vert_2^2 \rangle = \langle F(x), x-x^* \rangle, \quad \forall x \in U \setminus \{x^*\}.
    \end{equation*} 
\end{remark}

Therefore, if we have variational stability, we can use the squared Euclidean distance to the equilibrium as a Lyapunov function and use it to determine a basin of attraction. 

\paragraph{Region of Attraction} Using Corollary~\ref{cor:basin_attraction}, we can determine a region from which we always converge to the equilibrium. In this example, we have for instance $\Ucal_{c}$ with $c=\tfrac 1 9$, which corresponds to the ball around the equilibrium with radius $\tfrac 1 3$. 

\begin{figure}[h]
    \centering
    \includegraphics[width=\textwidth]{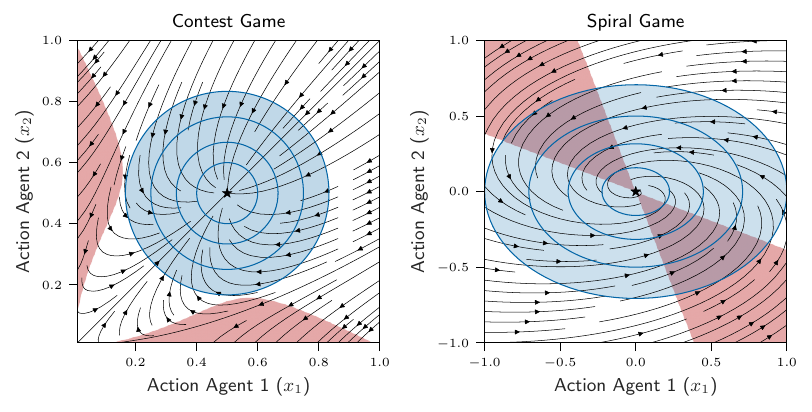}
    \caption{Gradient dynamics for contests and the spiral game. The red areas show violations of variational stability. The blue lines show the level sets of the respective Lyapunov function while the blue area represents the determined basin of attraction.}\label{fig:minty}
\end{figure}

\subsubsection{Spiral Game - Non-Trivial Lyapunov Function}
In another example, we want to show that variational stability might not even be satisfied locally, but we can still get (regional) convergence and determine a basin of attraction using the theory of Lyapunov functions.
To that end, we constructed a game which we call \textit{Spiral Game}. The continuous game has two players \(i \in \{1,2\}\) with feasible sets \(\Xcal_i = [-C, C]\) for some $C > 0$. The utility functions are given by 
\[u_1(x_1,x_2) = - \tfrac 1 2 x_1^2 - 4 x_1 x_2, \quad u_2(x_1,x_2) = - \tfrac 1 4 x_2^2 + x_1 x_2.\]

\paragraph{Nash Equilibrium} Again, the game is obviously concave and we can find an interior equilibrium by computing a $x^*$ such that $F(x^*) = 0$:
\[ F(x_1,x_2) = \begin{pmatrix} -x_1 - 4 x_2, \\ - \tfrac 1 2 x_2 + x_1 \end{pmatrix} = 0 \quad \Rightarrow \quad x^* = (0,0). \]

\paragraph{Local Stability} Using the Jacobian we can show again that the equilibrium is locally stable since the real parts of the eigenvalues are negative: 
\[J(x^*) = \begin{bmatrix}
- 1 & -4 \\
1 & -\frac{1}{2}
\end{bmatrix} \quad \Rightarrow \quad \lambda_{1,2} = - \dfrac{3}{4} \pm  \dfrac{3}{4} \sqrt{7} i. \]

\paragraph{Variational Stability} Interestingly, the game is not even locally variation stable. If we consider the points \(x = (x_1, -x_1)\), we can get arbitrarily close to the equilibrium but still violate the Minty condition: 
\[ \langle F(x), x-x^*\rangle = (- x_1 - 4 (-x_1))x_1 + (-\tfrac 1 2 (-x_1) + x_1)(-x_1) =  \tfrac 3 2 x_1^2 > 0\text{ for all } x_1 \neq 0. \] 

\paragraph{Lyapunov Function} As we have explained in Remark \ref{rem:vs_lyap}, (local) variational stability is stronger than Lyapunov stability, as it corresponds to a very specific form of a Lyapunov function. Allowing for more general functions, we can show that $V(x) = x_1^2 + 2x_2^2 $ is a Lyapunov function in a neighborhood of $x^*$.
Following Definition~\ref{def:lyapunov-function}, it is easy to see that \(V(x)\) is positive definite and for the derivative we have:
\[\dot V(x) = \nabla V(x)^TF(x) = 2x_1(-x_1 - 4 x_2) + 4x_2( \tfrac 1 2 x_2 + x_1) = -2(x_1 + x_2)^2 < 0, \quad \forall x \neq 0. \]
Note that we only consider a neighborhood of $x^*$ such that we do not have to consider projections at the boundary of $\Xcal$.
Since the Lyapunov conditions hold for all interior points, we just have to look for level sets that are contained within the feasible set to determine a \textit{region of attraction}. For this specific example with $\Xcal = [-1,1] \times [-1,1]$ a possible choice is $\Ucal_c$ with $c=1$.

If we are interested in global convergence or if we have an equilibrium on the boundary of the games action set, we have to consider projections. This is what we will do in the next section.

%%%%%%%%%%%%%%%%%%%%%%%%%%%%%%%%%%%%%%%%%%%%%%%%%%%%%%%%%%%%%%%%%%%%%%%%%%%%%%%%%%%%%%%%%%%%%%%%%%%%%%%%%%
% 	5  DYNAMICS OF GAMES WITH CONSTRAINTS		                                                         %
%%%%%%%%%%%%%%%%%%%%%%%%%%%%%%%%%%%%%%%%%%%%%%%%%%%%%%%%%%%%%%%%%%%%%%%%%%%%%%%%%%%%%%%%%%%%%%%%%%%%%%%%%%
\section{Dynamics of Games with Constraints}\label{sec:projected-dynamics}

In the previous chapter, we assumed that the players' strategies can evolve freely on an unbounded space $\Xcal$. Recall that we modeled the unconstrained dynamics by $\dot{x} = f(x) = F(x)$ with the game gradient $F$. 

Constraints and projected dynamical systems arise, among other cases, for finite games and their mixed extensions. In finite games, a player's action space is some finite set of possible choices. This scenario not only covers many of the standard game-theoretical examples; it also gains practical relevance when we consider discretized versions of continuous games. For many algorithms, this even becomes a necessity since they are only defined on discrete action sets. 

More specifically, many algorithms work on the mixed extension of a finite game, meaning that an algorithm comes to choose a probability distribution over its possible actions. Unfortunately, this restricts the space of possible strategies to the simplex, so a new complexity arises from this approach. We now need to deal with projected dynamics, and several approaches that we discussed earlier, like the linearization of the dynamics and the analysis of the Jacobian, are no longer valid.

Despite these complexities, there are several benefits connected to games in mixed extensions. \citet{nash1950equilibrium} showed that a (possibly mixed) Nash equilibrium has to exist.
Furthermore, the utility functions are now continuous in the mixed strategies, even if the continuous counterpart of a game has discontinuities in the payoff function. Auctions are a classical example of this, where the allocation of the item(s) is often non-smooth. 
Not only are the utility functions of mixed-extensions continuous, but they are also linear (and thus concave) in the strategy of an agent. Thus, the set of solutions to a Stampaccia-type variational inequality is equivalent to the Nash equilibria of these games.

Due to their practical relevance and abundance in many scenarios, we dedicate the following chapter to projected game dynamics. We start with continuous-time considerations before we dive deeper into game dynamics in discrete time.

\subsection{Games with Continuous-Time Dynamics}

Now, we consider projected dynamics in compact and convex spaces, how they can be modeled, and what impact this has on the stability analysis.
Examples of dynamics in continuous time include: 
\begin{enumerate}[label=\roman*)]
    \item \textit{Globally Projected Dynamical System (GPDS) \citep{friesz1994day}:}
    \[
        \dot{x} = \proj_{\Xcal}\bigl(x - \eta F(x_t)\bigr) - x
    \]
    where
    \[
    \proj_{\Xcal}(x) = \argmin_{\hat{x} \in \Xcal} \lVert \hat{x} - x \rVert.
    \]
    
    \item \textit{Locally Projected Dynamical System (LPDS) \citep{nagurney2012projected}:}
    \[
        \dot{x} = \proj_{T_{\Xcal}(x)}\bigl(x, F(x)\bigr)
    \]
    where the projection onto the tangent cone \(T_{\Xcal}(x)\) is defined by
    \[
        T_{\Xcal}(x) = \overline{\{v \in \mathbb{R}^n : \exists \lambda > 0 \text{ such that } x + \lambda v \in \Xcal\}},
    \]
    and
    \[
        \proj_{T_{\Xcal}(x)}(v) = \argmin_{z \in T_{\Xcal}(x)} \|v - z\|.
    \]
\end{enumerate}

Recall from Chapter \ref{sec:game-theory} that the solutions of the variational inequality $VI(\Xcal, F)$ correspond to the Nash equilibria of the game for concave utility functions. \citet{dupuis1993dynamical} showed that these solutions also coincide with the equilibrium points of LPDS and GPDS for convex polyhedra $\Xcal$ such as the probability simplex in the mixed extension of a finite game.
\citet{dupuis1993dynamical} also established the existence and uniqueness of the solution path of the ordinary differential equation (ODE). Using projected dynamical systems, an equilibrium can now be seen as an end-product of a dynamic process, rather than an isolated solution of a time-independent and static variational inequality problem \citep{pappalardo2002stability}. Unfortunately, even if the equilibrium of an unconstrained dynamical system is asymptotically stable, this might not be the case if we consider constraints, as the example in Figure \ref{fig:unstable} illustrates.

\begin{figure}[hpt!]
  \begin{center}  
  \begin{tikzpicture}[scale=1.0]
    
        % Axes (horizontal: X2, vertical: X1)
        \draw[->] (-3,0) -- (3,0) node[right] {$X_2$};
        \draw[->] (0,-2) -- (0,2) node[above] {$X_1$};
      
        % Define triangle vertices: Constraint region as triangle A--B--O
        \coordinate (A) at (-2,-1);
        \coordinate (B) at ( 2,-1);
        \coordinate (O) at (0,0);
        
        % Fill the triangle (the constrained region)
        \fill[gray!20] (A) -- (B) -- (O) -- cycle;
        
        % Draw triangle boundaries
        \draw[thick] (A) -- (B) -- (O) -- cycle;
        
        % Mark vertices with labels
        \fill (A) circle (2pt) node[left] {$A$};
        \fill (B) circle (2pt) node[right] {$B$};
        \fill (O) circle (2pt) node[below left] {$O$};
      
        % Draw a horizontally stretched spiral (blue) that converges to the origin.
        % The spiral is defined in polar form with
        % r(t) = 1.5*(1 - t/12.56), with t from 0 to 12.56 (approximately 4pi).
        % The x-coordinate is now multiplied by 3 to stretch further.
        \draw[thick, niceBlue, domain=0:12.56, samples=200, variable=\t]
          plot ({3*1.5*(1-\t/12.56)*cos(\t r)}, {1.5*(1-\t/12.56)*sin(\t r)});
          
        % Add a red arrow showing the gradient vector near where the spiral hits OB
        % Approximate point on OB where spiral crosses: near (1.2, -0.6)
        \draw[->, thick, niceRed] (1.5,-0.7) -- (2.0,-0.6)
          node[right] {\footnotesize $\nabla u$};          
      \end{tikzpicture}
      \caption{An example where the equilibrium point of the dynamical system $(0,0)$ is globally asymptotically stable without the constraint set $A-O-B$, but not if we consider the constraints. The projected dynamics at the edge $O-B$ would drag the dynamics away from the equilibrium (Figure inspired by \citet{nagurney2012projected}). }\label{fig:unstable}
      \end{center}
    \end{figure}
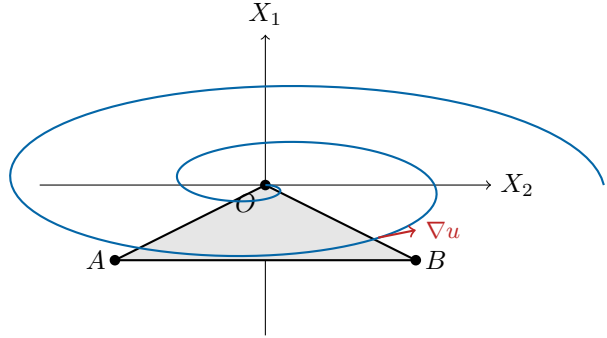

Under some conditions, \citet{dupuis1993dynamical} provide a convergence guarantee of the discrete-time projected dynamics to the continuous LPDS/GPDS as the learning rate $\eta$ approaches zero. These conditions include a convex constraint set as well as monotonicity of the game gradient. Unfortunately, even simple games such as the first-price auction might not be monotonic \citep{bichler2023convergence}. Another regularity condition demands that if the equilibrium is on the boundary, then the dynamics point inward of the feasible region, rather than along or out of the boundary. If the equilibrium is in the interior of the feasible region, then the monotonicity conditions are generally sufficient since regularity holds by default. If the equilibrium is on the boundary of the feasible region, then regularity is additionally required to ensure that the dynamics converge toward equilibrium. Their results can be summarized as follows:

\begin{theorem}[Solutions for LPDS \citep{nagurney2012projected}] \label{thm:nagurney}
Suppose that $x^*$ is an equilibrium point to $LPDS(F,\Xcal)$. Then
\begin{enumerate}[label=\roman*)]
\item if $F$ is locally pseudomonotone at $x^*$, then $x^*$ is a monotone attractor;
\item if $F$ is locally strictly pseudomonotone at $x^*$, then $x^*$ is a strictly monotone attractor;
\item if $F$ is locally strongly monotone at $x^*$, then $x^*$ is exponentially stable;  
\item if $F$ is pseudomonotone/strictly pseudomonotone/strongly monotone on $\Xcal$, $x^*$ is a global monotone attractor/strictly global monotone attractor/globally exponentially stable.
\end{enumerate}
\end{theorem}

\citet{zhang1996stability} discuss stability analysis of an adjustment process for oligopolistic market equilibrium modeled as a projected dynamical system. 

In contrast to this oligopoly game, the equilibria of a game can be at the boundary of a constraint or at a vertex. For example, the feasible strategy profiles of the mixed extension of a finite, normal-form game need to satisfy a probability simplex. Pure-strategy Nash equilibria are the vertices of the simplex. We neither know if the game is monotone nor can we assume that it satisfies the regularity conditions in \citet{nagurney2012projected}. The analysis of continuous-time dynamics needs to consider the discontinuity of gradient dynamics at the boundaries, a deviation from standard Lyapunov theory. Non-smooth Lyapunov theory provides a way to deal with such discontinuities at the boundaries \citep{cortes_discontinuous_2008, hauswirth_projected_2021}. In order to deal with the discontinuities, the related literature often draws on differential inclusions. A \textit{differential inclusion} is a generalization of a differential equation where the derivative of the state is not given by a single vector but rather is allowed to belong to a set, e.g., a set of subdifferentials. 
Many classical results from ordinary differential equations extend to differential inclusions. For example, Lyapunov stability theory, LaSalle's invariance principle, and various existence and uniqueness results have analogs in the differential inclusion setting \citep{brogliato2006equivalence, cortes_discontinuous_2008, filippov_differential_1988, bacciotti_stability_1999}.   

\citet{mertikopoulos2016learning} analyzed projected dynamics, more specifically mirror ascent dynamics, which are related to but different from the projected gradient ascent algorithm that we focus on in this paper. 
They take a different approach by letting players update their strategies according to cumulative payoff scores that are regularized by a penalty function. 
Nonsteep penalty functions (e.g., quadratic penalties) and projections ensure that the strategies remain viable but again introduce discontinuities (or non-smooth behavior). The authors deal with these instances by considering piecewise smooth (or Carathéodory) solutions. This approach guarantees that the dynamics are well-defined almost everywhere, ensuring convergence and stability even when trajectories can hit the boundary of the probability simplex. 
In mixed extension games of finite, normal-form games, \citet{mertikopoulos2016learning} state that the Nash equilibria coincide with the equilibrium points of the dynamics. Moreover, they show that all Lyapunov-stable states are Nash equilibria, and that strict Nash equilibria are asymptotically stable. 
Similarly, \citet{mertikopoulosLearningGamesContinuous2019} analyze general continuous games on convex action spaces by relating the convergence behavior of the dual projected dynamical system with the concept of variational stability. If a game is strictly monotone, then its unique Nash equilibrium is globally variationally stable.
By basing part of their analysis on the "Fenchel coupling" rather than solely on the local behavior of the gradient, the authors effectively bypass the difficulties encountered at the boundary. This duality-based method allows them to extend stability and convergence results to settings where the projection dynamics are at play. The Fenchel coupling essentially relates the primal (strategy) space and the dual (score or payoff) space so that one can track the evolution of the system even in the presence of nonsmooth behavior. \citet{mertikopoulos2018riemannian} generalizes these ideas, situating them within a broader geometric framework that unifies and extends many classical and contemporary models of evolutionary game dynamics.

\subsection{Games with Discrete-Time Dynamics}

% In this paper, we look at the dynamics induced by projected gradient ascent as a standard first-order algorithm. 
In the previous section, we used the continuous-time dynamics as an approximation of the dynamics arising from projected gradient ascent. Most algorithms are defined in discrete time, however, which is why we will focus on analyzing discrete-time dynamics in the remainder of the paper.
The dynamics of independent and projected gradient ascent algorithms will be described as
\begin{equation}
\label{eq:projected-dynamical-system-discrete-time}
    x_{t+1} = f(x_{t}) := \proj_{\Xcal}(x_t + \eta F(x_t)),
\end{equation}
where $\proj_{\Xcal}(x)$ is the projection of a point $x \in \R^n$ to $\Xcal$ and $F$ is again the game gradient. Assuming that the game gradient is continuous and $\Xcal$ is convex, which is the case for all discretized mixed-extension games, the dynamics function $f$ is continuous, too, since the projection is Lipschitz-continuous. We thus will assume continuity of $f$ in the following.

Below, we provide an overview of theorems that will enable us to analyze stability and the basin of attraction.
The results are based on the book by \citet{lasalle_stability_1986}, and we provide proofs for our versions in Appendix \ref{sec:proofs-stability-discrete-subspace-dynamics}. We note that all of these theorems are applicable to closed subsets of $\R^n$, including the simplex, by interpreting open sets, neighborhoods, and balls in the context of subspace topologies. We refer interested readers to Appendix \ref{sec:proofs-stability-discrete-subspace-dynamics}.

We begin with defining a Lyapunov function in this setting, which is given as the following.

\begin{definition}[Lyapunov function for discrete-time systems]
\label{def:discrete-time-Lyapunov-function}
    Let $\Ucal$ be any set in $\Xcal$. A function $V$ is a Lyapunov function for the discrete-time dynamical system (\ref{eq:projected-dynamical-system-discrete-time}) on $\Ucal$ if
    \begin{enumerate}[label=\roman*)]
        \item $V$ is continuous on $\Xcal$,
        \item $\Delta V(x) := V(f(x)) - V(x) \leq 0$ for all $x \in \Ucal$,
        \item $V(x^*) = 0$, and
        \item $V(x) > 0$ on $x \in B_\delta(x^*) \setminus \{x^*\}$ for some $\delta > 0$.
    \end{enumerate}
\end{definition}

Based on a Lyapunov function, one can verify stability or asymptotic stability under additional assumptions.

\begin{theorem}[Lyapunov stability for discrete-time systems]
    \label{thm:Lyapunov-stability-discrete-time-on-domain-D}
    Consider the discrete-time dynamical system (\ref{eq:projected-dynamical-system-discrete-time}) with continuous function $f$.
    If $V$ is a Lyapunov function for some neighborhood of $x^*$, then $x^*$ is a stable equilibrium.
\end{theorem}

\begin{theorem}[Asymptotic stability for discrete-time systems]
    \label{thm:asymptotic-stability-on-domain-D}
    If $V$ is a Lyapunov function such that $  \Delta V (x) < 0 $ for all $x \in \Ucal \setminus \{x^*\}$,
    then $x^*$ is asymptotically stable.
\end{theorem}

Note that, if $\Ucal = \Xcal$ in the above theorem, we get global asymptotic stability of $x^*$. If asymptotic stability cannot be established globally, the following corollary can be used to approximate the basin of attraction instead. Both results can be derived from LaSalle's invariance principle, which we provide as Theorem \ref{thm:LaSalle-invariance-principle-discrete-on-domain-D} in Appendix \ref{sec:proofs-stability-discrete-subspace-dynamics}.

\begin{corollary}[Basin of attraction for discrete-time systems]
\label{cor:basin-of-attraction-discrete}
    For an equilibrium $x^*$ that is asymptotically stable according to the Lyapunov function $V$, let 
    \begin{equation*}
        \Vcal = \{x^*\} \cup \{x : V(x) > 0, \Delta V(x) < 0\} \quad \text{and} \quad \Ucal_{c} = \{x: V(x) \leq c\}.
    \end{equation*}
    If $\Ucal_{c} \subseteq \Vcal$ and $\Ucal_{c}$ is bounded, $\Ucal_{c}$ is a subset of the basin of attraction.
\end{corollary}

In the following, we will apply the stability concepts above to different types of equilibria.

\subsection{Strict Nash Equilibria, Variational Stability, and Asymptotic Stability}

We start our analysis with strict Nash equilibria because they guarantee asymptotic stability, a result that we detail here.
Similar to the unconstrained continuous-time setting, there is a connection between variationally stable equilibria and the existence of a specific Lyapunov function (cf. Remark~\ref{rem:vs_lyap}). This connection allows us to get asymptotic stability of the equilibrium under the discrete-time dynamics \eqref{eq:projected-dynamical-system-discrete-time}.
But similar to fixed-point iterations of this form for variational inequalities, where we need \textit{strong} monotonicity to show convergence (reference), we need a stronger version of variational stability in our setting.

\begin{proposition} \label{prop:svs_implies_as}
    Let $G=(\Ncal, \Xcal, u)$ be a continuous game with an $L$-Lipschitz game gradient $F$. 
    If an equilibrium is \emph{strongly variationally stable}, i.e.,
    there exists a neighborhood $U$ of $x^*$ such that
    \begin{equation} \label{eq:strong_vs}
        \langle F(x), x - x^* \rangle \leq - \alpha \Vert x - x^* \Vert^2_2, \quad \forall x \in \Ucal
    \end{equation}
    for some $\alpha > 0$, then the equilibrium is asymptotically stable under the discrete-time dynamics \eqref{eq:projected-dynamical-system-discrete-time} with a step size $\eta < \tfrac{2 \alpha }{L^2} $.
\end{proposition}

\begin{proof}
    To apply Theorem~\ref{thm:asymptotic-stability-on-domain-D}, we have to find a Lyapunov function according to Definition~\ref{def:discrete-time-Lyapunov-function} with $\Delta V(x) < 0$. 
    Since the equilibrium is variationally stable, a natural candidate for a Lyapunov function is the squared Euclidean distance to the equilibrium $V(x):=\Vert x - x^* \Vert_2^2$ (cf. Remark~\ref{rem:vs_lyap}). 
    As all other properties of the Lyapunov function follow directly, it remains to show that $\Delta V(x)  < 0 $ for all $x \in \Ucal$ where $\Ucal$ is some neighborhood of $x^*$:
    \begin{align*}
        V(f(x)) - V(x) 
        &= \Vert \Pi_{\Xcal}(x + \eta F(x)) - x^* \Vert_2^2 - \Vert x - x^* \Vert_2^2 \\
        \intertext{Using the non-expansiveness of the projection on $\Xcal$ (compact \& convex) and the fact that $x^*$ is a fixed point under $f(\cdot)$ we get}
        &\leq \Vert \eta F(x_t) - \eta  F(x^*)  +  x - x^* \Vert_2^2 - \Vert x - x^* \Vert_2^2 \\
        &= \eta^2 \Vert F(x) - F(x^*) \Vert_2^2 + 2 \eta \langle F(x), x-x^* \rangle \\
        \intertext{With the  Lipschitz continuity of the game gradient and the strong variational stability of the equilibrium we have}
        &\leq \eta ^2 L^2 \Vert x - x^* \Vert^2 - 2 \eta \alpha \Vert x - x^* \Vert_2^2 
        = \eta (\eta L^2 - 2 \alpha) \Vert x - x^* \Vert_2^2\\
        &< 0 \quad \text{ for } \eta < \tfrac{2 \alpha }{L^2}, \quad \forall x \in \Ucal \setminus \{ x^* \}
    \end{align*}
\end{proof}

Let us use this result and analyze a well-known finite game, the Prisoner's Dilemma, to illustrate the case. 
Each player has two strategies: \( S_1 \) and \( S_2 \). The payoff matrices for Player 1 and Player 2 are:
\[
A_1 = \begin{bmatrix} 3 & 0 \\ 5 & 1 \end{bmatrix}, \quad A_2 = \begin{bmatrix} 3 & 5 \\ 0 & 1 \end{bmatrix}.
\]
The components of the game gradient of $F(x_1, x_2) = (F_1(x_1, x_2), F_2(x_1, x_2))$ are:
\[
F_1(x) = A_1 x_2, \quad F_2(x) = A_2^\top x_1.
\]
The unique Nash equilibrium is given by $x^* = (x_1^*, x_2^*)^T$ with $x_1^* = x_2^* = (0, 1)^T $.
It is easy to show that the equilibrium is strongly variationally stable:
\begin{align*}
    \langle F(x), x-x^*\rangle 
    &= \langle A_1 x_2, x_1 - x_1^* \rangle +  \langle A_2^T x_1, x_2 - x_2^* \rangle \\
    &= 6 x^{1}_{1} x^{1}_{2} + 5 x^{1}_{1} x^{2}_{2} - 5 x^{1}_{1} + 5 x^{2}_{1} x^{1}_{2} + 2 x^{2}_{1} x^{2}_{2} - x^{2}_{1} - 5 x^{1}_{2} - x^{2}_{2} \\
    \intertext{Using $x^{2}_{i} = 1 - x^{1}_{i}$ for $i = 1,2$ we can simplify the computations and get}
    &=  - 2 x^{1}_{1} x^{1}_{2} - x^{1}_{1} - x^{1}_{2} \leq - \left[(x^{1}_{1})^2 + (x^{1}_{2})^2 \right] \\
    &= - \tfrac 1 2  \left[ (x^{1}_{1})^2 + (1-x^{1}_{2})^2 + (x^{1}_{2})^2 + (1-x^{2}_{2})^2 \right] = -\tfrac 1 2 \Vert x-x^* \Vert_2^2
\end{align*}
Therefore, the equilibrium in the prisoners' dilemma is strongly variationally stable with $\alpha = \tfrac 1 2$ for all $x \in \Xcal$ and thereby globally asymptotically stable.

In general, not all equilibria of a game must be globally variationally stable, and it can be tedious to verify that a given equilibrium is strongly variationally stable. In these cases, the strictness of the equilibrium is enough to provide us with local properties.
From \citet[Proposition~5.2]{mertikopoulosLearningGamesContinuous2019} we already know that all strict Nash equilibria of a finite game are locally variationally stable and vice versa. Following \citep{cohen2017hedge} this result can be extended and one can show that all strict Nash equilibria are locally strongly variationally stable, which implies asymptotic stability under the discrete-time dynamics with sufficiently small step size:

\begin{proposition} \label{prop:sne_implies_as}
    Given a finite normal-form game. If $x^*$ is a strict equilibrium, then there exists a neighborhood of $x^*$, where the equilibrium is asymptotically stable under projected gradient ascent \eqref{eq:projected-dynamical-system-discrete-time} with sufficiently small step size.
\end{proposition}

\begin{proof}
    Proposition~1 in \citet{cohen2017hedge} states that for every strict equilibrium $x^*$ in a finite normal-form game, there exists a neighborhood $\Ucal$ of $x^*$ such that $\langle F(x), x-x^* \rangle \leq - \alpha \Vert x - x^* \Vert_1 $ for some $\alpha > 0$. 
    Since all components of $x-x^*$ are in $[-1,1]$ we also have $ \Vert x - x^* \Vert \geq \Vert x - x^* \Vert_2^2 $ which makes the equilibrium (locally) strongly variationally stable (cf. Equation~\eqref{eq:strong_vs}) and allows us to apply Proposition~\ref{prop:svs_implies_as} giving us the result.
    Note that the game gradient in mixed extensions of finite normal-form games is always Lipschitz continuous.
\end{proof}

Consider the "Battle of Sexes" game, which is a two-player game with the actions "opera" and "football", and which has the following payoff matrices:

\begin{equation*}
    \begin{aligned}
        A_1 = \begin{bmatrix}
            3 & 0 \\ 0 & 2
        \end{bmatrix},
        &&
        A_2 = \begin{bmatrix}
            2 & 0 \\ 0 & 3
        \end{bmatrix}
    \end{aligned}
\end{equation*}

This game has three equilibria, two of them in pure strategies: 
\begin{equation*}
    \left(\begin{bmatrix} 0 \\ 1 \end{bmatrix}, \begin{bmatrix} 0 \\ 1 \end{bmatrix}\right), \quad
    \left(\begin{bmatrix} 1 \\ 0 \end{bmatrix}, \begin{bmatrix} 1 \\ 0 \end{bmatrix}\right), \quad 
    \left(\begin{bmatrix} \nicefrac{2}{5} \\ \nicefrac{3}{5} \end{bmatrix}, \begin{bmatrix} \nicefrac{3}{5} \\ \nicefrac{2}{5} \end{bmatrix}\right)
\end{equation*}

Both pure-strategy equilibria are strict, i.e., unilateral deviations strictly decrease the utility, and, by Proposition \ref{prop:sne_implies_as}, we get that they must be (locally) asymptotically stable under the dynamics (\ref{eq:projected-dynamical-system-discrete-time}) for a sufficiently small $\eta$. This statement also extends to other gradient-based learning schemes, such as dual averaging \citep{mertikopoulosLearningGamesContinuous2019}. 

\subsection{Basin of Attraction}

Both pure equilibria are asymptotically stable in the "Battle of Sexes" game, so neither can be globally asymptotically stable. However, the local nature of this statement does not tell us anything about the regions from which the dynamics will converge to either of these equilibria. In this simple game, we can visually determine the basin of attraction, which is shown in Figure \ref{fig:battle-of-sexes_basin-of-attraction}\footnote{The figures actually display only two dimensions of the four-dimensional state space $\Xcal$. However, since the simplex in $\R^2$ is a one-dimensional manifold, this is enough to fully represent the dynamics here.}. We have already noted the relation between variational stability and asymptotic stability for continuous-time dynamics in Remark \ref{rem:vs_lyap}, which is why we display the region where variational stability holds in Figure \ref{fig:battle-of-sexes_variational-stability}. Below, we will make efforts to estimate parts of the basin of attraction by means of Corollary \ref{cor:basin_attraction} for the discrete-time dynamical system (\ref{eq:projected-dynamical-system-discrete-time}) with a learning rate $\eta = 0.05$. This procedure can also be used in games where a visual determination of the basin of attraction is infeasible.

\begin{figure}[h]
    \centering
    \begin{subfigure}{0.45\textwidth}
        \includegraphics[width=\textwidth]{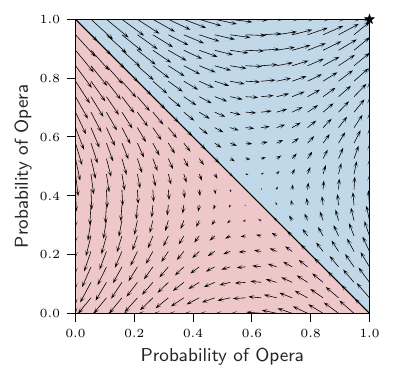}
        \subcaption{Basin of attraction with discrete dynamics.}
        \label{fig:battle-of-sexes_basin-of-attraction}
    \end{subfigure}
    \hfill
    \begin{subfigure}{0.45\textwidth}
        \includegraphics[width=\textwidth]{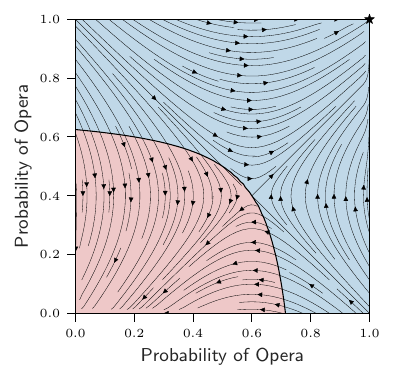}
        \subcaption{Variational stability with continuous dynamics.}
        \label{fig:battle-of-sexes_variational-stability}
    \end{subfigure}
    \caption{Basin of attraction and variational stability in the "Battle of Sexes" game in two dimensions.
    The figures display the basin of attraction and the region where variational stability holds for the equilibrium $x_1 = (1, 0)$, $x_2 = (1, 0)$ in blue. The left picture uses discrete-time dynamics, while the right picture displays the gradient field of the continuous-time dynamics. For the discrete-time dynamics, the arrows indicate real movements in the state space, $x_t \to x_{t+1}$, while the continuous-time dynamics are represented by their time-derivatives $\dot{x}$.
    }
    \label{fig:battle-of-sexes_boa-and-vs}
\end{figure}

Following Propositions \ref{prop:svs_implies_as} and \ref{prop:sne_implies_as}, we know that $V(x) = \norm{x - x^*}_2^2$ is a Lyapunov function that allows us to establish (local) asymptotic stability. The condition $\Delta V < 0$ is implied by strong variational stability, $\langle F(x) , x - x^* \rangle \leq - \alpha \norm{x - x^*}_2^2$, on a (potentially small) neighborhood of the equilibrium. Following Corollary \ref{cor:basin_attraction}, we determine the set
\begin{equation*}
    \Vcal = \{x^*\} \cup \{x : V(x) > 0, ~ \Delta V(x) = V(\proj_{\Xcal} (x + \eta F(x))) - V(x) < 0 \}.
\end{equation*}
Figure \ref{fig:battle-of-sexes_region-V} shows a numerically obtained visualization of $\Vcal$. Observe that this region, just like the variationally stable set, is not a subset of the basin of attraction! Due to the projections in the definition of $\Vcal$, an analytical analysis is difficult. We can instead use the non-expansiveness of the projection operator to determine an upper bound like we did in the proof of Proposition \ref{prop:svs_implies_as}:
\begin{align*}
    \overline{\Delta V}(x) &:= \eta^2 \lVert F(x) - F(x^*) \rVert_2^2 - 2 \eta \langle F(x) , x - x^* \rangle \geq \Delta V(x)  \\
    \bar{\Vcal} &:= \{x^*\} \cup \{x : V(x) > 0, ~ \tilde{\Delta V}(x) < 0 \} \subseteq \Vcal
\end{align*}
Now, $\bar{\Vcal}$ is a subset of $\Vcal$, but the projection is no longer needed to compute the set. For two-player games, $\overline{\Delta V}$ is simply a quadratic function, which makes the analysis easier. We visualize $\Vcal$ in Figure \ref{fig:battle-of-sexes_region-V} and $\bar{\Vcal}$ in Figure \ref{fig:battle-of-sexes_region-V-bar}. Figure \ref{fig:battle-of-sexes_region-V-bar} also shows the differences between $\bar{\Vcal}$, $\Vcal$, and the variationally stable set. For $\eta \to 0$, $\Vcal$ will approach the variationally stable set and $\bar{\Vcal}$ will approach $\Vcal$. We choose a constant $c = 0.8$ for which $\Ucal_c = \{x : V(x) \leq c\} \subset \R^4$ is a subset of $\bar{\Vcal}$ (and thus of $\Vcal$). Due to Corollary \ref{cor:basin_attraction}, $\Ucal_c$ is a subset of the basin of attraction, which we can also verify visually in this case.

\begin{figure}[h]
    \centering
    \begin{subfigure}{0.45\textwidth}
        \includegraphics[width=\textwidth]{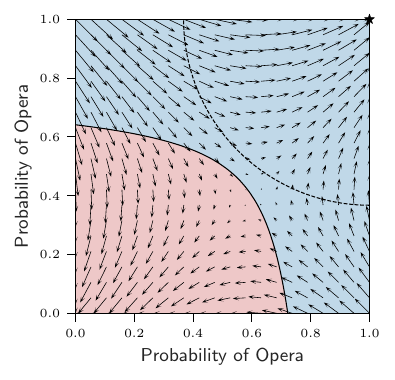}
        \subcaption{Sets $\Vcal$ and $\Ucal_c$.}
        \label{fig:battle-of-sexes_region-V}
    \end{subfigure}
    \hfill
    \begin{subfigure}{0.45\textwidth}
        \includegraphics[width=\textwidth]{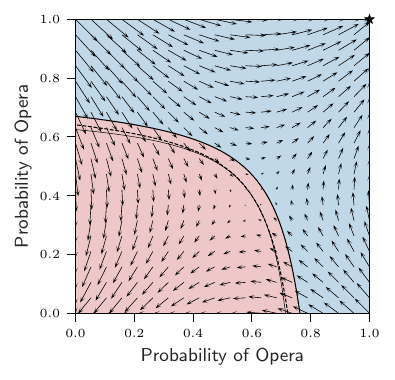}
        \subcaption{Set $\bar{\Vcal}$.}
        \label{fig:battle-of-sexes_region-V-bar}
    \end{subfigure}
    \caption{Estimating the basin of attraction with Corollary \ref{cor:basin_attraction}. Figure a) displays the set $\Vcal = \{x^*\} \cup \{x: V(x) > 0, \Delta V(x) < 0\}$ for the equilibrium $x_1 = (1, 0)$, $x_2 = (1, 0)$ in blue and the dashed boundary of $\Ucal_c$. Figure b) shows the set $\bar{\Vcal}$ where the upper bound $\overline{\Delta V}$, derived by using the non-expansiveness of the projection operator, is lower than zero. Figure b) also shows the exact boundary of $\Vcal$ (dashed line) and the boundary for variational stability (thin line).}
    \label{fig:battle-of-sexes_subsets-of-boa}    
\end{figure}

\subsection{Mixed and Weak Nash Equilibria} 

So far, we have assumed that an equilibrium of interest is strict. In this section, we will provide examples showcasing that all kinds of dynamics can evolve around mixed and weak Nash equilibria.

We already discussed that games with only mixed Nash equilibria cannot be stable and attracting under certain no-regret dynamics \citep{flokas2020no}. A classical example of a bimatrix game is the matching pennies game, where two players can either choose heads ($H$) or tails ($T$). One player wins \$1 if both choose the same action simultaneously, the other wins if they do not. 
%Let player 1 choose $H$ with probability $p$ and player 2 choose $H$ with probability $q$. 
While \textit{strict} monotonicity is only a sufficient condition for convergence and a violation is no proof for non-convergence, it is interesting to see that Matching Pennies is a monotonous, but not a strictly monotonous game, and that this strictness in the inequality makes a big difference. 

\[
A_1 = \begin{bmatrix}
        1 & -1 \\
        -1 & 1 \\
    \end{bmatrix}, \qquad A_2 = -A_1^T
\]
% The expected utility is given by $u_i(x_1, x_2) = x_1^T A_i x_2$, which gives us the game gradient $F(x) = (A_1 x_2, A_2^T x_1)$ where $x = (x_1, x_2)$ is the strategy profile. 
It is easy to see that the mixed strategy $x^*=(x_1^*, x_2^*)$ with $x_1^* = x_2^* = (\frac 1 2, \tfrac 1 2)$ is a Nash equilibrium, since $F(x^*) = 0$.
If we check for (strict) monotonicity, we can see that the inequality is zero for all pairs of strategy profiles, which implies monotonicity but not strict monotonicity:
\[
\langle F(x) - F(y), x-y \rangle = \langle 
\begin{pmatrix} A_1 x_2 \\ A_2^T x_1 \end{pmatrix} - \begin{pmatrix} A_1 y_2 \\ A_2^T y_1 \end{pmatrix},
\begin{pmatrix} x_1 \\ x_2\end{pmatrix} - \begin{pmatrix} y_1 \\ y_2\end{pmatrix} = (x_1 - y_1)^T\underbrace{(A_1 - A_2^T)}_{=0}(x_2-y_2) = 0.
\]
With gradient dynamics (\ref{eq:projected-dynamical-system-discrete-time}), the equilibrium is not stable, and trajectories starting nearby will spiral away from $x^*$; see Figure \ref{fig:matching-pennies}. This indicates that, contrary to strong monotonicity, monotonicity is insufficient to guarantee asymptotic stability.

\begin{figure}[h]
    \centering
    \includegraphics[width=0.5\linewidth]{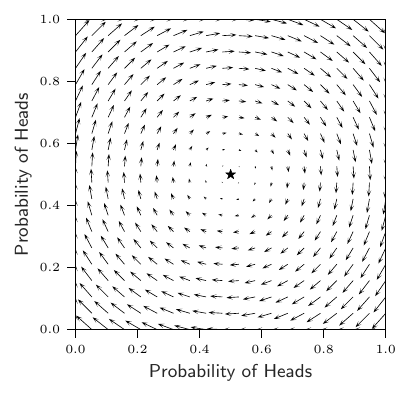}
    \caption{Discrete-time dynamics of "Matching Pennies" ($\eta = 0.05$).}
    \label{fig:matching-pennies}
\end{figure}

The "Matching Pennies" game is not a singular example. Other games with only mixed Nash equilibria, such as the Shapley game \citep{shapley1963some}, can also not be learned with standard gradient dynamics, and algorithms cycle. 

This observation is not limited to truly mixed equilibria. \textit{Weak} pure-strategy Nash equilibria are not necessarily sufficient for convergence \citep{mertikopoulosLearningGamesContinuous2019}, too. For example, the following 2-player game from \citep{milionis2022nash} has several pure Nash equilibria (and a continuum of mixed equilibria), and it cycles:

\[
    A_1 = \begin{bmatrix}
        1 & 0 & -1 \\
        -1 & 0 & -1 \\
        1 & 0 & -2
    \end{bmatrix}, \qquad
    A_2 = \begin{bmatrix}
        1 & -1 & 1 \\
        0 & 0 & 0 \\
        -1 & -1 & 2
    \end{bmatrix}
\]

Already \citet{hart2003uncoupled} showed that there are games where no uncoupled learning dynamics converge to the Nash equilibrium. 
The next game is a two-player game which has only one weak PNE and it doesn't converge with gradient dynamics as well.

\[
    A_1 = \begin{bmatrix}
        1 & 2 & 3 \\
        0 & 2 & 0 \\
        3 & 2 & 1 
    \end{bmatrix}, \qquad
    A_2 = \begin{bmatrix}
        3 & 0 & 1 \\
        1 & 2 & 2 \\
        1 & 0 & 3 
    \end{bmatrix}
\]

%http://www.dklevine.com/archive/refs4475.pdf
%https://en.wikipedia.org/wiki/Fictitious_play
%http://www.ma.huji.ac.il/~hart//papers/uncoupl.pdf

In the examples above, mixed and weak Nash equilibria were not asymptotically stable. One might think that a game with a single strict Nash equilibrium $x^*$ (and potentially other mixed Nash equilibria that are not asymptotically stable) will thus converge to $x^*$ for almost all initializations. This, however, is not the case. In the following, we provide an example that shows that, while strictness is a necessary condition for global convergence \citep{giannou2021survival, mertikopoulosLearningGamesContinuous2019}, it is not a sufficient condition.

\begin{proposition} \label{prop:unique_spne_no_global_convergence}
Gradient-based learning algorithms do not necessarily converge to a unique strict Nash equilibrium for all initial conditions.
\end{proposition}
The proof can be found in Appendix~\ref{app:proofs}. 
It is based on the analysis of the dynamics in a specific game, namely an extended version of matching pennies with parameters $r, q \in \mathbb{R}^+$, $q > 1$, and payoff matrices
\begin{equation*}
	\begin{aligned}
		A_1 = \begin{bmatrix}
		r & -q & -q \\
		- \nicefrac{q}{2} & 1 & -1 \\
		- \nicefrac{q}{2} & -1 & 1
	\end{bmatrix},
	& &
	A_2 = \begin{bmatrix}
		r & - \nicefrac{q}{2} & - \nicefrac{q}{2} \\
		- q & -1 & 1 \\
		- q & 1 & -1
	\end{bmatrix}
	\end{aligned}.
\end{equation*}
Since the analysis in discrete-time is challenging, we prove the result for continuous-time dynamics.

\subsection{Limitations and Research Challenges}
Let us now examine an example that illustrates the limitations of using Lyapunov functions, as discussed in this paper.

We consider a two-player Cournot oligopoly model, where each agent (firm) \(i \in \{1, 2\}\) chooses a quantity $a_i$ to supply to the market. The market price $p(a)$ depends on the total supply, implying that each firm's decision affects the price faced by all participants. The utility of firm $i$ is given by 
\begin{equation}
    u_i(a) = a_i \, p(a) - a_i \, c_i = a_i \,  (b_0 - \sum \limits_{i \in \Ical} b_i a_i ) - a_i \, c_i,
\end{equation}
where $c_i \geq 0$ represents the marginal production costs, $b_i > 0$ captures the firms price-setting power, and $b_0$ is some constant.
We focus on a simple symmetric case with parameters $b_0 = 1$, $b_1 = b_2 = \tfrac 1 2$, $c_1 = c_2 = 0$.

\paragraph{The Problem of Discretization} 
When the actions are continuous, i.e., $\Acal_i = [0, 1]$, the resulting game is strongly monotone and has a unique equilibrium at $a^* = (a_1^*, a_2^*) = (\tfrac 2 3, \tfrac 2 3)$. 
The equilibrium is thereby exponentially stable under the locally projected dynamical system (cf. Theorem~\ref{thm:nagurney}). 
Furthermore, from the theory on variational inequalities, it is also well known that fixed-point iterations - corresponding to the discrete-time dynamics \eqref{eq:projected-dynamical-system-discrete-time} - converge to the equilibrium in strongly monotone settings.
However, if we instead consider discrete actions, such as $ \Acal_i = \{ \tfrac 0 7, \tfrac 1 7, \dots, 1 \}$ and analyze the mixed extension of this game, this property no longer holds.
In this discretized setting, the game has three pure equilibria $a^* \in \{(\tfrac 4 7, \tfrac 5 7), (\tfrac 5 7, \tfrac 4 7), (\tfrac 5 7, \tfrac 5 7)\}$. 
The presence of multiple equilibria implies that the game cannot be monotone, as monotonicity guarantees uniqueness. 
Moreover, all three pure equilibria are weak, which rules out (strong) variational stability.
Consequently, the local convergence guarantees provided by our previously introduced tools no longer apply in this setting.

\paragraph{Stability of Sets}
This issue is characteristic of many discretized games—including models of oligopolies, contests, and auctions -- where multiple, potentially non-strict and neighboring equilibria arise.
Classical Lyapunov-based asymptotic stability typically implies convergence to a single point and requires strict decrease of a Lyapunov function along the system's trajectories.
However, in settings with multiple equilibria, convergence to a set is often more appropriate and meaningful.
In our example, one could attempt to construct a Lyapunov function that demonstrates convergence to a face of the probability simplex containing all three equilibria.
This relaxation is common in control theory and systems with nonsmooth dynamics (e.g., due to projections), where the concept of stability for sets is frequently analyzed via differential inclusions, a generalization of differential equations. 
\citet{bacciotti_stability_1999} developed techniques for constructing Lyapunov functions for sets rather than points. Their solution and similar approaches are usually based on generalized versions of \citeauthor{lasalle_extensions_1960}'s invariance theorem.
The application of this theory to game dynamics describes a challenging research direction.

An alternative approach is to consider broader equilibrium concepts, similar to (coarse) correlated equilibria, such as the recently introduced semicoarse correlated equilibria (SCCEs) (\( \text{NE} \subseteq \text{CE} \subseteq \text{SCCE} \subseteq \text{CCE} \)).
Importantly, SCCEs characterize the outcome of projected gradient dynamics and can be formulated as a linear program similar to the CCE polytope, but with additional constraints  \citep{ahunbay2025semicoarsecorrelatedequilibrialpbased}.
The dual of this linear program can be interpreted as a Lyapunov function that certifies convergence of projected gradient dynamics to elements of the SCCE set.

Overall, techniques to prove convergence to a set of equilibria rather than an isolated and strict Nash equilibrium is a veritable research challenge, which is important to fully characterize the agentic markets that emerged recently.

%%%%%%%%%%%%%%%%%%%%%%%%%%%%%%%%%%%%%%%%%%%%%%%%%%%%%%%%%%%%%%%%%%%%%%%%%%%%%%%%%%%%%%%%%%%%%%%%%%%%%%%%%%
% 	6	OUTLOOK									                                                         %
%%%%%%%%%%%%%%%%%%%%%%%%%%%%%%%%%%%%%%%%%%%%%%%%%%%%%%%%%%%%%%%%%%%%%%%%%%%%%%%%%%%%%%%%%%%%%%%%%%%%%%%%%%
\section{Summary and Outlook}

This survey provides an overview of the interplay between learning algorithms and game dynamics in markets populated by autonomous and learning agents. By integrating concepts from online learning, game theory, and dynamical systems, we lay out a framework for analyzing how self-interested agents interact and adapt over time in complex market environments with learning agents.

In this survey, we assume a static environment where agents play the same game repeatedly. Think about thousands of customers who are served by two sellers of luxury cars each day. In such markets, the same suppliers compete repeatedly over time. If inefficient price cycles arise in such envrionments, this phenomenon can also arise in a more dynamic scenario where supply and demand changes over time. On the other hand, if learning algorithms implement an efficient equilibrium outcome in such static models, then price cycles or algorithmic collusion might be less of a concern in the field. 

The analysis of game dynamics with learning agents is still a research area with many open questions. New ideas are required to characterize games that do not exhibit strict Nash equilibria. Future developments will also address dynamic market environments and real-world environments as they can be found in display advertising auctions or on retail platforms. 

%%%%%%%%%%%%%%%%%%%%%%%%%%%%%%%%%%%%%%%%%%%%%%%%%%%%%%%%%%%%%%%%%%%%%%%%%%%%%%%%%%%%%%%%%%%%%%%%%%%%%%%%%%
% 	ACKNOWLEDGMENTS								                                                         %
%%%%%%%%%%%%%%%%%%%%%%%%%%%%%%%%%%%%%%%%%%%%%%%%%%%%%%%%%%%%%%%%%%%%%%%%%%%%%%%%%%%%%%%%%%%%%%%%%%%%%%%%%%

\section*{Acknowledgments}
This project was funded by the Deutsche Forschungsgemeinschaft (DFG, German Research Foundation) - GRK 2201/2 - Project Number 277991500 and BI 1057/9.

%%%%%%%%%%%%%%%%%%%%%%%%%%%%%%%%%%%%%%%%%%%%%%%%%%%%%%%%%%%%%%%%%%%%%%%%%%%%%%%%%%%%%%%%%%%%%%%%%%%%%%%%%%
%	BIBLIOGRAPHY				                                                           			     %
%%%%%%%%%%%%%%%%%%%%%%%%%%%%%%%%%%%%%%%%%%%%%%%%%%%%%%%%%%%%%%%%%%%%%%%%%%%%%%%%%%%%%%%%%%%%%%%%%%%%%%%%%%

\bibliographystyle{informs2014} 
%\bibliography{literature.bib}

%%% Uncomment this line and comment out the ``thebibliography'' section below to use the external .bib file (using bibtex)

%%% Uncomment this section and comment out the \bibliography{references} line above to use inline references.
\begin{appendix}

%%%%%%%%%%%%%%%%%%%%%%%%%%%%%%%%%%%%%%%%%%%%%%%%%%%%%%%%%%%%%%%%%%%%%%%%%%%%%%%%%%%%%%%%%%%%%%%%%%%%%%%%%%
%	APPENDIX				                                                           			     %
%%%%%%%%%%%%%%%%%%%%%%%%%%%%%%%%%%%%%%%%%%%%%%%%%%%%%%%%%%%%%%%%%%%%%%%%%%%%%%%%%%%%%%%%%%%%%%%%%%%%%%%%%%

\section{Stability of Discrete-Time Systems with Subspace Topologies}
\label{sec:proofs-stability-discrete-subspace-dynamics}

Most versions of Lyapunov theorems and LaSalle's invariance principle are formulated on $\R^n$. 
This section establishes results for dynamics
\begin{equation}
    \label{eq:discrete-dynamical-system-on-domain-D}
    x(t+1) = f(x(t)), \quad x(0) = x_0 \in \Xcal
\end{equation}
with an arbitrary closed domain $\Xcal \subseteq \R^n$ and a continuous function $f: \Xcal \to \Xcal $. Note that $\Xcal$ can also be a measure zero set w.r.t. the Lebesgue measure on $\R^n$, e.g., the simplex on $\R^n$. We will build on the theorems provided by \citet{lasalle_stability_1986} and modify them slightly to our needs.

Let us briefly recall the definitions of a subspace topology and a neighborhood, as this provides the context for our discussion on stability. Note that an open neighborhood in the subspace topology (based on $\Xcal$) may not be open in the original topology (based on $\R^n$). 

\begin{definition}[Subspace topology\footnote{from \href{https://en.wikipedia.org/wiki/Subspace_topology}{Wikipedia}}]
    Given a topological space $(X, \tau)$ and a subset $S$ of $X$, the subspace topology on $S$ is defined by
    \begin{equation*}
        \tau_S = \{S \cap U \vert U \in \tau \}.
    \end{equation*}
    That is, a subset of $S$ is open in the subspace topology if and only if it is the intersection of $S$ with an open set in $(X, \tau)$. 
\end{definition}

\begin{definition}[Neighborhood\footnote{from \href{https://en.wikipedia.org/wiki/Neighbourhood_(mathematics)}{Wikipedia}}]
    If $X$ is a topological space and $p$ is a point in $X$, then a neighborhood of $p$ is a subset $V$ of $X$ that includes an open set $U$ containing $p$:
    \begin{equation*}
        p \in U \subseteq V \subseteq X.
    \end{equation*}
\end{definition}

In the following, we consider the subspace topology $(\Xcal, \tau_{\Xcal})$ with $\tau_{\Xcal} = \{\Xcal \cap U \vert U \in \tau\}$ being the collection of open sets.
Based on subspace topologies, balls are sets of the form $B_\delta(x) = \{y \in \Xcal: \norm{y - x} \leq \delta\}$. 

We now continue with Lyapunov theorems and LaSalle's invariance principle, starting with a relaxed definition of a Lyapunov function that we already provided in Definition \ref{def:discrete-time-Lyapunov-function}.

\begin{definition}[Lyapunov function \citep{lasalle_stability_1986}]
    Let $\Ucal$ be any set in $\Xcal$. A function $V$ is a Lyapunov function for the discrete-time dynamical system (\ref{eq:discrete-dynamical-system-on-domain-D}) on $\Ucal$ if
    \begin{enumerate}[label=\roman*)]
        \item $V$ is continuous on $\Xcal$ and
        \item $\Delta V(x) := V(f(x)) - V(x) \leq 0$ for all $x \in \Ucal$.
    \end{enumerate}
\end{definition}

This definition essentially drops the assumption of positive definiteness, which is not required for LaSalle's invariance principle (see below). The notions of Lyapunov stability, however, will explicitly state this assumption to complement their versions from the main part of the paper.

Using Lyapunov functions, we restate Theorem \ref{thm:Lyapunov-stability-discrete-time-on-domain-D}, which a version of Lyapunov's well-known second method.

\begin{theorem}[Lyapunov stability for discrete-time systems \citep{lasalle_stability_1986}]
    % \label{thm:Lyapunov-stability-discrete-time-on-domain-D}
    Consider the discrete-time dynamical system (\ref{eq:discrete-dynamical-system-on-domain-D}) with continuous function $f$.
    If $V$ is a Lyapunov function for some neighborhood of $x^*$ and $V$ is positive definite with respect to $x^*$, that is
    \begin{enumerate}[label=\roman*)]
        \item $V(x^*) = 0$, and
        \item $V(x) > 0$ on $x \in B_\delta(x^*) \setminus \{x^*\}$ for some $\delta > 0$,
    \end{enumerate}
    then $x^*$ is a stable equilibrium.
\end{theorem}

Note that, in the main part of this article, we assumed that any Lyapunov function is positive definite, which is a more restrictive definition that allows us to make the same statements.

\begin{proof}[Proof of Theorem \ref{thm:Lyapunov-stability-discrete-time-on-domain-D}]
    We may take $\eta$ so small that $V(x) > 0$ and $\Delta V \leq 0$ for $x \in B_\eta(x^*)$. Let $\varepsilon > 0$ be given; there is no loss in generality in taking $0 < \eta < \varepsilon$. 
    Let $m = \min \{V(x): \norm{x - x^*} = \eta, ~ x \in \Xcal \}$ be the minimum value of $V$ on the circle or arc drawn around $x^*$ with radius $\eta$. 
    $m$ is positive since we are taking the minimum of a positive continuous function over a compact set. \\
    Let $\Ucal = \{x : V(x) < m/2\}$ and $\Ucal_0$ the connected component of $\Ucal$ which contains $x^*$. Both $\Ucal$ and $\Ucal_0$ are open, and $\Ucal_0 \subset B_\eta(x^*)$ due to continuity of $V$. If $x_0 \in \Ucal_0$, then $\Delta V(x_0) \leq 0$, so $V(f(x_0)) \leq V(x_0) < m/2$, and thus $x_0 \in \Ucal$. 
    Since $x_0$ and $x^*$ are in the same component of $\Ucal$, so are $f(x^*) = x^*$ and $f(x_0)$ due to continuity of $f$. 
    Thus $\Ucal_0$ is an open positively invariant set containing $x^*$ and contained in $B_{\varepsilon}(x^*)$.  \\
    Since $V$ is continuous, there is a $\delta > 0$ such that $B_\delta(x^*) \subset \Ucal_0$. 
    So if $x_0 \in B_\delta$, then $x(t) \in \Ucal_0 \subset B_{\varepsilon}(x^*)$.
\end{proof}

We proceed with a variant of LaSalle's invariance principle. We later use this theorem to prove asymptotic stability by Lyapunov functions and our corollary characterizing the basin of attraction.

\begin{theorem}[LaSalle's invariance principle for discrete-time systems]
    \label{thm:LaSalle-invariance-principle-discrete-on-domain-D}
    If $V$ is a Lyapunov function in the positively invariant, compact set $\Ucal \subseteq \Xcal$ for the discrete-time system (\ref{eq:discrete-dynamical-system-on-domain-D}) with a continuous function $f$, then $x(n) \to \Mcal$ where $\Mcal$ is the largest positively invariant set contained in the set $\Ecal = \{x \in \R^n: V(f(x)) - V(x) = 0\} \cap \Ucal$.
\end{theorem}

\begin{proof}
    % Before starting with the proof, we need to state the following proposition about $\omega$-limit sets of our dynamics.
    % \begin{proposition}
    %     \label{prop:omega-limit-set-for-discrete-time}
    %     If $x(t)$ is bounded for all $t \in \mathbb{Z}^+$, then $\omega(x_0)$ is a non-empty, compact, positively invariant set. Moreover $x(t) \to \omega(x_0)$ as $t \to \infty$.
    % \end{proposition}
    % For a proof of this proposition, refer to \citet{lasalle_stability_1986}. We now proceed with the proof of LaSalle's invariance principle.
    If $x_0 \in \Ucal$, we have $x(t) \in \Ucal$, so $\omega = \omega(x_0) \neq \emptyset$, $\omega \subseteq \Ucal$, and $x(t)$ tends to $\omega$ by Proposition 1.10 in \citep{lasalle_stability_1986}. 
    Now $V(x(t))$ is non-increasing and bounded below, so $V(x(t)) \to c$. 
    If $y \in \omega$, then there is a subsequence $t_k$ such that $x(t_k) \to y$, so $V(x(t_k)) \to V(y)$ (since $V$ is continuous) or $V(y) = c$. 
    Thus, $V(\omega) = c$ or $\omega \subseteq V^{-1}(c)$.
     Also, since $V(\omega) = c$ and $\omega$ is positively invariant, $\Delta V(\omega) = 0$. So $x(t) \to \omega \subseteq \{x \in \R^n: \Delta V = 0\} \cap \bar{\Ucal} \cap V^{-1}(c)$. 
     Since $\omega$ is positively invariant, we have $\omega \subseteq \Mcal$. 
\end{proof}

Now, we restate and proof Theorem \ref{thm:asymptotic-stability-on-domain-D}.

\begin{theorem}[Asymptotic stability for discrete-time systems \citep{lasalle_stability_1986}]
    % \label{thm:asymptotic-stability-on-domain-D}
    If $V$ is a Lyapunov function such that
    \begin{enumerate}[label=\roman*)]
        \item $V(x^*) = 0$, 
        \item $V(x) > 0$ on $x \in B_\delta(x^*) \setminus \{x^*\}$ for some $\delta > 0$, and
        \item $V(f(x)) - V(x) < 0$ for all $x \in \Ucal \setminus \{x^*\}$,
    \end{enumerate}
    then $x^*$ is asymptotically stable.
\end{theorem}

\begin{proof}[Proof of Theorem \ref{thm:asymptotic-stability-on-domain-D}]
    Since $\Delta V > 0$ on a neighborhood of $x^*$, the point $x^*$ is stable by Theorem \ref{thm:Lyapunov-stability-discrete-time-on-domain-D}.
    From the proof of Theorem \ref{thm:Lyapunov-stability-discrete-time-on-domain-D}, there is an arbitrary small neighborhood $\Ucal_0$ of $x^*$ which is positively invariant. We can make $\Ucal_0$ so small that $V(x) > 0$ and $\Delta V(x) < 0$ for $x \in \Ucal \setminus \{x^*\}$. Given any $x_0 \in \Ucal_0$, we have, by the invariance principle, that $x(t)$ tends to the largest invariant set in $\Ucal_0 \cap \{\Delta V (x) = 0\} = \{x^*\}$ since $- \Delta V$ is positive definite. Thus, $x^*$ is asymptotically stable. 
\end{proof}

LaSalle's invariance theorem can again be used to approximate the basin of attraction, as stated in Corollary \ref{cor:basin-of-attraction-discrete}:
\begin{corollary}[Basin of attraction for discrete-time systems]
% \label{cor:basin-of-attraction-discrete}
    For an equilibrium $x^*$ that is asymptotically stable according to the Lyapunov function $V$, let 
    \begin{equation*}
        \Vcal = \{x^*\} \cup \{x : V(x) > 0, \Delta V(x) < 0\} \quad \text{and} \quad \Ucal_{c} = \{x: V(x) \leq c\}
    \end{equation*}
    If $\Ucal_{c}l$ and $\Ucal_{c}$ is bounded, $\Ucal_{c}$ is a subset of the basin of attraction.
\end{corollary}

\begin{proof}[Proof of Corollary \ref{cor:basin-of-attraction-discrete}]
    If $\Ucal_{c} \subseteq \Vcal$, $\Ucal_{c}$ is invariant. LaSalle's invariance principle (Theorem \ref{thm:LaSalle-invariance-principle-discrete-on-domain-D}) tells us that $x(n) \to \Mcal$ with $\Mcal \subseteq \Ecal = \{x \in \R^n: V(f(x)) - V(x) = 0\} = \{x^*\}$.
\end{proof}

\section{Additional Proofs} \label{app:proofs}

\begin{proof}[Proof of Proposition~\ref{prop:variational-stability-by-pure-deviations}]
	Using that the expected utility is linear in the agent's own argument (cf. Definition~\ref{def:mixed_extension}, it is easy to see that 
%    Let us first establish the equivalence of both terms. Assume arbitrary $x, x' \in \Delta(\Acal)$.	
	\begin{equation*}
		\langle F(x), x - x' \rangle = \sum_i \langle F_i(x), x_i \rangle - \langle F_i(x), x'_i \rangle = \sum_{i = 1}^{n} \left(u_i(x) - u_i(x_i', x_{-i})\right).
	\end{equation*}
	% The last equivalence follows from the tensor-product definition of the reward for mixed-extension games:
	% \begin{equation*}
 %    \begin{aligned}
	% 	& && u_i(x) = \sum_{a \in \Acal} x_i(a_i) \cdot x_{-i}(a_{-i}) \cdot u_i(a_i, a_{-i}) \\
 %        &\implies && F_{i}^k(x) = \nabla_{x_i^k} u_i(x) = \sum_{a_{-i} \in \Acal_{-i}} x_{-i}(a_{-i}) \cdot u_i(a_i^k, a_{-i}) \\
 %        &\implies && u_i(x) = \langle F_i(x), x_i \rangle
 %    \end{aligned}
	% \end{equation*}
	The "implies" direction of our theorem simply follows by the fact that all pure strategies are also mixed strategies. If the variational stability condition holds for the latter, it must also hold for the former.\\
	Next, let us show the "implied by" direction of our theorem. To that end, assume that $\langle F(x), x - x^* \rangle < 0$ for all \emph{pure} $x \neq x^*$. We want to establish that $\langle v(x), x - x^* \rangle < 0$ for \emph{all} (potentially non-pure) $x \neq x^*$.
	\begin{equation*}
		\begin{split}
			\langle F(x), x - x^* \rangle 
			&= \sum_{i = 1}^{n} u_i(x) - u_i(x_i^*, x_{-i}) \\
			&= \sum_{i = 1}^{n} \E_{a \sim x}\left[u_i(a) - u_i(a_i^*, a_{-i}) \right]\\
%			&= \sum_{i = 1}^{n} \sum_{a_{-i} \in \Acal_{-i}} x_{-i}(a_{-i}) \cdot \left( \left(\sum_{a_i \in \Acal_i} x_i(a_i) \cdot u_i(a_i, a_{-i})\right) - \left(\sum_{a_i \in \Acal_i} x_i(a_i) \cdot u_i(a_i^*, x_{-i})\right) \right) \\
			&= \sum_{i = 1}^{n} \sum_{a \in \Acal} x(a) \cdot \left(u_i(a) - u_i(a^*, a_{-i})\right) \\
			&= \sum_{a \in \Acal} x(a) \underbrace{\sum_{i = 1}^{n} \left(u_i(a) - u_i(a^*, a_{-i})\right)}_{< 0} < 0
		\end{split}
	\end{equation*}
\end{proof}

\begin{proof}[Proof of Proposition~\ref{prop:unique_spne_no_global_convergence}]
We proof the statement by a counterexample. Consider the following game, which we call "Extended Matching Pennies", with parameters $r, q \in \mathbb{R}^+$, $q > 1$, and payoff matrices
\begin{equation*}
	\begin{aligned}
		A_1 = \begin{bmatrix}
		r & -q & -q \\
		- \nicefrac{q}{2} & 1 & -1 \\
		- \nicefrac{q}{2} & -1 & 1
	\end{bmatrix},
	& &
	A_2 = \begin{bmatrix}
		r & - \nicefrac{q}{2} & - \nicefrac{q}{2} \\
		- q & -1 & 1 \\
		- q & 1 & -1
	\end{bmatrix}
	\end{aligned}.
\end{equation*}
We denote the strategy components of player 1 and player 2 by $x_1 = (x_1^1, x_1^2, x_1^3)$ and $x_2 = (x_2^1, x_2^2, x_2^3)$.
This game is a matching pennies game for the last two actions of each player, but it also allows them to play a third action. The strict equilibrium of this game is achieved when both players play the first action and receive a reward of $r$. We want to analyze the convergence behavior of the game towards or away from this equilibrium. For this, we consider the (projected) gradients of the game and how they evolve for the component $x_1^1$.

The game gradient of player 1 is
\begin{equation*}
	F_1 = \nabla_{x_1} u_1 = A_1 x_2 = \begin{bmatrix}
		r x_2^1 - q(x_2^2 + x_2^3) \\
		- \nicefrac{q}{2} \cdot x_2^1 + x_2^2 - x_2^3 \\
		- \nicefrac{q}{2} \cdot x_2^1 - x_2^2 + x_2^3
	\end{bmatrix}
\end{equation*}

Assume that $x_2^1 = 0$. In this case, player 1 will not have an incentive to increase $x_1^1$. If randomly initialized, this scenario happens with probability zero. Next, we analyze the setting if $x_2^1 = \epsilon > 0$. Assuming all strategy components are non-zero, the LPDS projection of the gradient $F$ to the simplex can be described by a matrix product. Let $n = (1, 1, 1)^T$ be the normal vector of the simplex. Then the projected gradient is
\begin{equation*}
	\Pi_{T_{\Xcal}}(F_1) = (I - n (n^T n)^{-1} n^T) v_1 = \begin{bmatrix}
		\nicefrac{2}{3} r x_2^1 - q(x_2^2 + x_2^3) + \frac{q}{3}\\
		(\nicefrac{r}{3} - \nicefrac{q}{2}) x_2^1 + x_2^2 - x_2^3 + \frac{q}{3} \\
		(\nicefrac{r}{3} - \nicefrac{q}{2}) x_2^1 - x_2^2 + x_2^3 + \frac{q}{3}
	\end{bmatrix}
\end{equation*}

Let us analyze under which conditions the first component, $x_1^1$, increases.
\begin{equation*}
	\begin{split}
		\frac{2}{3} r x_2^1 - q(x_2^2 + x_2^3) + \frac{q}{3} &> 0 \\
		\frac{2}{3} r \epsilon - q(1 - \epsilon) + \frac{q}{3} &> 0 \\
		(\frac{2}{3} r + q) \epsilon &> \frac{2}{3} q \\
		\epsilon &> \frac{\frac{2}{3} q}{\frac{2}{3} r + q}
	\end{split}
\end{equation*}
A symmetric argument can be made for the utility gradient of player 2 and its change in component $x_2^1$. Thus, if the initial point is chosen such that $x_2^1 \leq \frac{\nicefrac{2}{3} q}{\nicefrac{2}{3} r + q}$ and $x_1^1 \leq \frac{\nicefrac{2}{3} q}{\nicefrac{2}{3} r + q}$, both agents will never increase their weight on action 1, meaning we will not converge to the strict Nash equilibrium. Since $\epsilon > 0$, such initializations occur with non-zero probability under most initialization methods.
\end{proof}

\end{appendix}

\end{document}